\documentclass[a4paper]{article}
\usepackage[hidelinks]{hyperref}
\usepackage{mythm}\title{A Categorical Theory of Patches}\author{Samuel Mimram \and Cinzia Di Giusto}\date{CEA, LIST\footnote{This work was partially supported by the French project ANR-11-INSE-0007 REVER.}}
\usepackage[english]{babel}
\usepackage[utf8]{inputenc}
\usepackage{headers}
\usepackage{macro}
\usepackage{wrapfig}

\newcommand{\vxym}[1]{\vcenter{\xymatrix@R=2ex@C=2ex{#1}}}

\newcommand{\qcong}{\quad\cong\quad}
\renewcommand{\graph}{\mathcal{G}}

\newcommand{\da}[1]{{\downarrow}{#1}}

\newenvironment{proposition*}[1]{\vspace{-\lastskip}\par \addvspace{.6pc
    plus .2pc minus .1pc}  \noindent {\bf #1} \begin{it}}{\end{it}\par\addvspace{.6pc
    plus .2pc minus .1pc}}

\begin{document}\maketitle


  %
  %
  
  \begin{abstract}
    When working with distant collaborators on the same documents, one often
    uses a version control system, which is a program tracking the history of
    files and helping importing modifications brought by others as patches. The
    implementation of such a system requires to handle lots of situations
    depending on the operations performed by users on files, and it is thus
    difficult to ensure that all the corner cases have been correctly
    addressed. Here, instead of verifying the implementation of such a system,
    we adopt a complementary approach: we introduce a theoretical model, which
    is defined abstractly by the universal property that it should satisfy, and
    work out a concrete description of it. We begin by defining a category of
    files and patches, where the operation of merging the effect of two
    coinitial patches is defined by pushout. Since two patches can be
    incompatible, such a pushout does not necessarily exist in the category, which raises the
    question of which is the correct category to represent and manipulate files
    in conflicting state. We provide an answer by investigating the free
    completion of the category of files under finite colimits, and give an
    explicit description of this category: its objects are finite sets labeled by lines
    equipped with a transitive relation and morphisms are partial functions respecting
    labeling and relations.
  \end{abstract}

\section{Introduction}

It is common nowadays, when working with distant collaborators on the same files
(multiple authors writing an article together for instance), to use a program
which will track the history of files and handle the operation of importing
modifications of other participants. These software called
\emph{version control systems} (\VCS for short), like \git or Darcs,
implement two main operations. When a user is happy with the changes it has
brought to the files it can record those changes in a \emph{patch} (a file
coding the differences between the current version and the last recorded
version) and \emph{commit} them to a server, called a \emph{repository}. The
user can also \emph{update} its current version of the file by importing new
patches added by other users to the repository and applying the corresponding
modifications to the files. One of the main difficulties to address here is that there is no
global notion of ``time'': patches are only partially ordered. For
instance consider a repository with one file~$A$ and two users~$u_1$
and~$u_2$. Suppose that $u_1$ modifies file~$A$ into~$B$ by committing a
patch $f$, which is then imported by $u_2$, and then~$u_1$ and $u_2$
concurrently modify the file~$B$ into $C$ (\resp $D$) by committing a patch $g$
(\resp $h$). The evolution of the file is depicted on the left and the partial
ordering of patches in the middle:
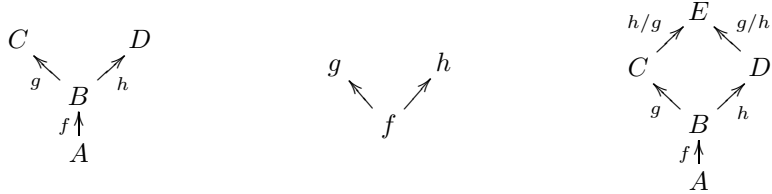
\[
\vxym{
  C&&D\\
  &B\ar[ul]^g\ar[ur]_h\\
  &A\ar[u]^f
}
\qquad\qquad\qquad
\vxym{
  g&&h\\
  &\ar[ul]f\ar[ur]&
}
\qquad\qquad\qquad
\vxym{
  &E&\\
  C\ar[ur]^{h/g}&&\ar[ul]_{g/h}D\\
  &B\ar[ul]^g\ar[ur]_h\\
  &A\ar[u]^f
}
\]
Now, suppose that $u_2$ imports the patch $g$ or that $u_1$ imports the patch
$h$. Clearly, this file resulting from the \emph{merging} of the two patches
should be the same in both cases, call it~$E$. One way to compute this file, is
to say that there should be a patch $h/g$, the \emph{residual} of $h$ after~$g$,
which transforms $C$ into~$E$ and has the ``same effect'' as $h$ once $g$ has
been applied, and similarly there should be a patch $g/h$ transforming $D$ into
$E$. Thus, after each user has imported changes from the other, the evolution
of the file is as pictured on the right above.
In this article, we introduce a category~$\linear$ whose objects are files and
morphisms are patches. Since residuals should be computed in the most general
way, we formally define them as the arrows of pushout cocones, \ie the square in
the figure on the right should be a pushout.

However, as expected, not every pair of coinitial morphisms have a pushout in
the category~$\linear$: this reflects the fact that two patches can be
conflicting (for instance if two users modify the same line of a
file). Representing and handling such conflicts in a coherent way is one of the
most difficult part of implementing a \VCS (as witnessed for instance by the
various proposals for Darcs: mergers, conflictors, graphictors,
etc.~\cite{darcs-theory}). In order to be able to have a representation for
all conflicting files, we investigate the free completion of the
category~$\linear$ under all pushouts, this category being
denoted~$\concurrent$, which corresponds to adding all conflicting files to the
category, in the most general way as possible. This category can easily be
shown to exist for general abstract reasons, and one of the main contributions of
this work is to provide an explicit description by applying the theory of
presheaves. This approach paves the way towards the implementation of a
\VCS whose correctness is deduced from universal categorical properties.
%
%

\medskip\noindent\emph{Related work.}
The Darcs community has investigated a formalization of patches based on
commutation properties~\cite{darcs-theory}. \emph{Operational
  transformations} tackle essentially the same issues by axiomatizing the
notion of residual patches~\cite{ressel1996integrating}. In both cases, the fact
that residual should form a pushout cocone is never explicitly stated,
excepting in informal sentences saying that ``$g/f$ should have the same effect
as $g$ once $f$ has been applied''. We should also mention another interesting
approach to the problem using inverse semigroups
in~\cite{jacobson2009formalization}.
%
Finally, Houston has proposed a category with pushouts, similar to ours, in
order to model conflicting files~\cite{bosker}, see
Section~\ref{sec:extensions}.

%

\medskip\noindent\emph{Plan of the paper.}
We begin by defining a category~$\linear$ of files and patches in
Section~\ref{sec:cat}. Then, in Section~\ref{sec:conflicting-cat}, we abstractly
define the category~$\concurrent$ of conflicting files obtained by free finite
cocompletion. Section~\ref{sec:cocompl} provides a concrete description of the
construction in the simpler case where patches can only insert lines. We give some 
concrete examples in Section~\ref{sec:examples} and adapt the framework to the
general case in Section~\ref{sec:extensions}. We conclude in
Section~\ref{sec:future-works}.


\section{Categories of files and patches}
\label{sec:cat}
In this section, we investigate a model for a simplified \VCS: it handles only
one file and the only allowed operations are insertion and deletion of lines
(modification of a line can be encoded by a deletion followed by an insertion).
We suppose fixed a set $\lines=\set{a,b,\ldots}$ of \emph{lines} (typically
words over an alphabet of characters). A \emph{file}~$A$ is a finite sequence of
lines, which will be seen as a function $A:[n]\to\lines$ for some number of
lines $n\in\N$, where the set~$[n]=\set{0,1,\ldots,n-1}$ indexes the lines of
the files. For instance, a file $A$ with three lines such that $A(0)=a$,
$A(1)=b$ and $A(2)=c$ models the file $abc$. Given $a\in \lines$, we sometimes
simply write $a$ for the file $A:[1]\to\lines$ such that $A(0)=a$.
A morphism between two files $A:[m]\to\lines$ and $B:[n]\to\lines$ is an
injective increasing partial function $f:[m]\to[n]$ such that $\forall
i\in[m],B\circ f(i)=A(i)$ whenever $f(i)$ is defined. Such a morphism is called
a \emph{patch}.

\begin{definition}
  \label{def:linear}
  The category $\linear$ has files as objects and patches as morphisms.
\end{definition}

Notice that the category~$\linear$ is strictly monoidal with
$[m]\otimes[n]=[m+n]$ and for every file $A:[m]\to\lines$ and $B:[n]\to\lines$,
$(A\otimes B)(i)=A(i)$ if $i<m$ and $(A\otimes B)(i)=B(i-m)$ otherwise, the unit
being the empty file $I:[0]\to\lines$, and tensor being defined on morphisms in
the obvious way. The following proposition shows that patches are generated by
the operations of inserting and deleting a line:

\begin{proposition}
  \label{prop:presimp-free-moncat}
  The category~$\linear$ is the free monoidal category containing~$\lines$ as
  objects and containing, for every line $a\in\lines$, morphisms $\eta_a:I\to a$
  (insertion of a line $a$) and $\varepsilon_a:a\to I$ (deletion of a line $a$)
  such that $\varepsilon_a\circ\eta_a=\id_I$ (deleting an inserted line amounts
  to do nothing).
\end{proposition}

\begin{example}
  The patch corresponding to transforming the file $abc$ into $dadeb$, by
  deleting the line~$c$ and inserting the lines labeled by $d$ and $e$, is
  modeled by the partial function $f:[3]\to[5]$ such that $f(0)=1$ and $f(1)=4$
  and $f(2)$ is undefined. Graphically,
  
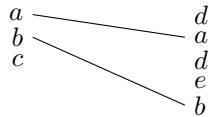
\[
  \begin{tikzpicture}[xscale=2.,yscale=0.6]
    \draw (0,2) node[left]{$a$};
    \draw (0,1.5) node[left]{$b$};
    \draw (0,1) node[left]{$c$};
    \draw (1,2) node[right]{$d$};
    \draw (1,1.5) node[right]{$a$};
    \draw (1,1) node[right]{$d$};
    \draw (1,.5) node[right]{$e$};
    \draw (1,0) node[right]{$b$};
    \draw (0,2) -- (1,1.5);
    \draw (0,1.5) -- (1,0);
  \end{tikzpicture}
  \]
  The deleted line is the one on which $f$ is not defined and the inserted lines
  are those which are not in the image of~$f$. In other words, $f$ keeps track of
  the unchanged lines.
\end{example}

In order to increase readability, we shall consider the particular case
where~$\lines$ is reduced to a single element. In this \emph{unlabeled} case,
the objects of~$\linear$ can be identified with integers (the labeling function
is trivial), and Proposition~\ref{prop:presimp-free-moncat} can be adapted to
achieve the following description of the category, see also~\cite{maclane:cwm}.

\begin{proposition}
  \label{prop:presimp-free-cat}
  If $\lines$ is reduced to a singleton, the category~$\linear$ is the free
  category whose objects are integers and morphisms are generated by $s_i^n:n\to
  n+1$ and $d_i^n:n+1\to n$ for every $n\in\N$ and $i\in[n+1]$ (respectively
  corresponding to insertion and deletion of a line at $i$-th position), subject
  to the relations
  \begin{equation}
    \label{eq:simpl-ax}
    s_i^{n+1}s_j^n=s_{j+1}^{n+1}s_i^n
    \qquad\qquad
    d_i^ns_i^n=\id_{n}
    \qquad\qquad
    d_i^nd_j^{n+1}=d_j^nd_{i+1}^{n+1}
  \end{equation}
  whenever $0\leq i\leq j<n$.
\end{proposition}

\noindent
We will also consider the subcategory $\tlinear$ of $\linear$, with same
objects, and \emph{total} injective increasing functions as morphisms. This
category models patches where the only possible operation is the insertion of
lines: Proposition~\ref{prop:presimp-free-moncat} can be adapted to show
that~$\tlinear$ is the free monoidal category containing morphisms $\eta_a:I\to
a$ and, in the unlabeled case, Proposition~\ref{prop:presimp-free-cat} can be
similarly adapted to show that it is the free category generated by morphisms
$s_i^n:n\to n+1$ satisfying $s_i^{n+1}s_j^n=s_{j+1}^{n+1}s_i^n$ with $0\leq
i\leq j<n$.

\section{Towards a category of conflicting files}
\label{sec:conflicting-cat}
Suppose that~$A$ is a file which is edited by two users, respectively applying
patches $f_1:A\to A_1$ and $f_2:A\to A_2$ to the file. For instance,
\begin{equation}
  \label{eq:no-conflict}
  \begin{matrix}
    a&c&c&b
  \end{matrix}
  \quad\xleftarrow{f_1}\quad
  \begin{matrix}
    a&b
  \end{matrix}
  \quad\xrightarrow{f_2}\quad
  \begin{matrix}
    a&b&c&d
  \end{matrix}
\end{equation}
Now, each of the two users imports the modification from the other one. The
resulting file, after the import, should be the smallest file containing both
modifications on the original file: $accbcd$.
It is thus natural to state that it should be a pushout of the
diagram~\eqref{eq:no-conflict}.
Now, it can be noticed that not every diagram in~$\linear$ has a pushout. For
instance, the diagram
\begin{equation}
  \label{eq:no-pushout}
  \begin{matrix}
    a&c&b
  \end{matrix}
  \quad\xleftarrow{f_1}\quad
  \begin{matrix}
    a&b
  \end{matrix}
  \quad\xrightarrow{f_2}\quad
  \begin{matrix}
    a&d&b
  \end{matrix}
\end{equation}
does not admit a pushout in~$\linear$. In this case, the two patches~$f_1$
and~$f_2$ are said to be \emph{conflicting}.

In order to represent the state of files after applying two conflicting patches,
we investigate the definition of a category~$\concurrent$ which is obtained by
completing the category~$\linear$ under all pushouts. Since, this completion
should also contain an initial object (\ie the empty file), we are actually
defining the category~$\concurrent$ as the free completion of $\linear$ under
finite colimits: recall that a category is finitely cocomplete (has all finite
colimits) if and only if it has an initial object and is closed under
pushouts~\cite{maclane:cwm}. Intuitively, this category is obtained by adding
files whose lines are not linearly ordered, but only partially ordered, such as
on the left of
\begin{equation}
  \label{eq:ex-conflict}
  \vxym{
    &\ar[dl]a\ar[dr]&\\
    c\ar[dr]&&\ar[dl]d\\
    &b&
  }
  \qquad\qquad\qquad\qquad\qquad
  \textbf{\texttt{\tiny
    \begin{tabular}{l}
      a\\[-1ex]
      <<<<<<< HEAD\\[-1ex]
      c\\[-1ex]
      =======\\[-1ex]
      d\\[-1ex]
      >>>>>>> 5c55\ldots\\[-1ex]
      b
    \end{tabular}
  }}
\end{equation}
which would intuitively model the pushout of the diagram~\eqref{eq:no-pushout}
if it existed, indicating that the user has to choose between~$c$ and~$d$ for
the second line. Notice the similarities with the corresponding textual notation
in \git on the right.
The name of the category~$\linear$ reflects the facts that its objects are files
whose lines are linearly ordered, whereas the objects of $\concurrent$ can be
thought as files whose lines are only partially ordered. More formally, the
category is defined as follows.

\begin{definition}
  \label{def:free-cocompl}
  The category~$\concurrent$ is the \emph{free finite conservative cocompletion}
  of~$\linear$: it is (up to equivalence of categories) the unique finitely
  cocomplete category together with an embedding functor
  $y:\linear\to\concurrent$ preserving finite colimits, such that for every
  finitely cocomplete category~$\C$ and functor $F:\linear\to\C$ preserving finite colimits, there exists,
  up to unique isomorphism, a unique functor $\tilde F:\concurrent\to\C$
  preserving finite colimits and satisfying $\tilde F\circ y=F$:
  \[
  \vxym{
    \linear\ar[d]_y\ar[r]^F&\C\\
    \concurrent\ar@{.>}[ur]_{\tilde F}
  }
  \]
\end{definition}

\noindent
Above, the term \emph{conservative} refers to the fact that we preserve colimits
which already exist in~$\linear$ (we will only consider such completions here).
The ``standard'' way to characterize the category~$\concurrent$, which always exists, is to use the following folklore
theorem, often attributed to Kelly~\cite{kelly1982basic,adamek1994locally}:

\begin{theorem}
  \label{thm:cons-cocomp}
  The conservative cocompletion of the category~$\linear$ is equivalent to the
  full subcategory of~$\hat\linear$ whose objects are presheaves which preserve
  finite limits, \ie the image of a limit in $\linear^\op$ (or equivalently a
  colimit in $\linear$) is a limit in $\Set$ (and limiting cones are transported
  to limiting cones). The finite conservative cocompletion~$\concurrent$ can be
  obtained by further restricting to presheaves which are finite colimits of
  representables.
\end{theorem}

\begin{example}
  \label{ex:one-cocompl}
  The category $\FSet$ of finite sets and functions is the conservative
  cocompletion of the terminal category~$\one$.
\end{example}

\noindent
We recall that the category $\hat\linear$ of \emph{presheaves} over~$\linear$,
is the category of functors $\linear^\op\to\Set$ and natural transformations
between them.
%
The \emph{Yoneda functor} $y:\linear\to\hat\linear$ defined on objects
$n\in\linear$ by $yn=\linear(-,n)$, and on morphisms by postcomposition,
provides a full and faithful embedding of $\linear$ into the corresponding
presheaf category, and can be shown to corestrict into a functor
$y:\linear\to\concurrent$~\cite{adamek1994locally}. A presheaf of the form $yn$
for some $n\in\linear$ is called \emph{representable}.

Extracting a concrete description of the category~$\concurrent$ from the above
proposition is a challenging task, because we a priori need to characterize
firstly all diagrams admitting a colimit in~$\linear$, and secondly all
presheaves in $\hat\linear$ which preserve those diagrams. This paper introduces
a general methodology to build such a category. In particular, perhaps a bit
surprisingly, it turns out that we have to ``allow cycles'' in the objects of
the category~$\concurrent$, which will be described as \emph{the category whose objects are
  finite sets labeled by lines together with a transitive relation and  morphisms
  are partial functions respecting labels and relations}.

\section{A cocompletion of files and insertions of lines}
\label{sec:cocompl}
In order to make our presentation clearer, we shall begin our investigation of
the category~$\concurrent$ in a simpler case, which will be generalized in
Section~\ref{sec:extensions}: we compute the free finite cocompletion of the
category $\tlinear$ (patches can only insert lines) in the case where the set of
labels is a singleton. \emph{To further lighten notations, in this section, we
  simply write $\linear$ for this category}.

We sometimes characterize the objects in $\linear$ as finite
colimits of objects in a subcategory~$\graph$ of~$\linear$. This
category~$\graph$ is the full subcategory of $\linear$ whose objects are~$1$
and~$2$: it is the free category on the graph
$\xymatrix{1\ar@<+.5ex>[r]\ar@<-.5ex>[r]&2}$, the two arrows being $s^1_0$ and
$s^1_1$. The category $\hat\graph$ of presheaves over~$\graph$ is the category
of \emph{graphs}: a presheaf $P\in\hat\graph$ is a graph with $P(1)$ as
vertices, $P(2)$ as edges, the functions $P(s_1^1)$ and $P(s_0^1)$ associate to
a vertex its source and target respectively, and morphisms correspond to usual
morphisms of graphs. We denote by $x\twoheadrightarrow y$ a path going from a
vertex $x$ to a vertex $y$ in such a graph. The inclusion functor $I:\graph\to\linear$
induces, by precomposition, a functor $I^*:\hat\linear\to\hat\graph$. The image
of a presheaf in~$\hat\linear$ under this functor is called its \emph{underlying
  graph}. By well known results about presheaves categories, this functor admits a
right adjoint $I_*:\hat\graph\to\hat\linear$: given a graph $G\in\hat\graph$,
its image under the right adjoint is the presheaf $G_*\in\hat\linear$ such that
for every $n\in\N$, $G_*(n+1)$ is the set of paths of length $n$ in the
graph~$G$, with the expected source maps, and $G_*(0)$ is reduced to one
element.

Recall that every functor $F:\C\to\D$ induces a \emph{nerve functor}
$N_F:\D\to\hat\C$ defined on an object $A\in\C$ by
$N_F(A)=\D(F-,A)$~\cite{mac1992sheaves}.
Here, we will consider the nerve $N_I:\linear\to\hat\graph$ associated to the
inclusion functor $I:\graph\to\linear$. An easy computation shows that the image
$N_I(n)$ of $n\in\linear$ is a graph with $n$ vertices, so that its objects are
isomorphic to $[n]$, and there is an arrow $i\to j$ for every $i,j\in[n]$ such
that $i<j$. For instance,
\[
N_I(3)=\vxym{0\ar[r]\ar@/_/[rr]&1\ar[r]&2}
\qquad\qquad\qquad
N_I(4)=\vxym{0\ar[r]\ar@/^/[rr]\ar@/_/[rrr]&1\ar[r]\ar@/^/[rr]&2\ar[r]&3}
\]
It is, therefore, easy to check that this embedding is full and faithful, \ie
morphisms in $\linear$ correspond to natural transformations in~$\hat\graph$.
Moreover, since~$N_I(1)$ is the graph reduced to a vertex and $N_I(2)$ is the
graph reduced to two vertices and one arrow between them, every graph can be
obtained as a finite colimit of the graphs $N_I(1)$ and $N_I(2)$ by ``gluing
arrows along vertices''. For instance, the initial graph $N_I(0)$ is the colimit
of the empty diagram, and the graph $N_I(3)$ is the colimit of the diagram
\begin{equation*}
  \vcenter{
    \xymatrix@R=1ex@C=4ex{
      &N_I(2)&&N_I(2)\\
      N_I(1)\ar[drr]_{N_I(s_1)}\ar[ur]^{N_I(s_1)}&&\ar[ul]_{N_I(s_0)}N_I(1)\ar[ur]^{N_I(s_1)}&&\ar[ul]_{N_I(s_0)}\ar[dll]^{N_I(s_0)}N_I(1)\\
      &&N_I(2)\\
    }
  }
\end{equation*}
which may also be drawn as on the left of
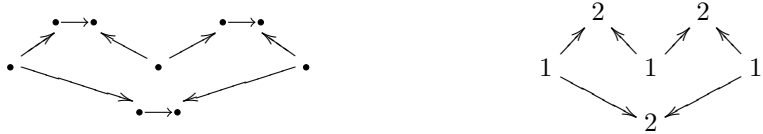
\newcommand{\anode}{{\begin{tikzpicture}\ssnode{(0,0)};\end{tikzpicture}}}
\newcommand{\twonodes}{{\begin{tikzpicture}\ssnode{(0,0)};\ssnode{(.5,0)};\end{tikzpicture}}}
\newcommand{\anedge}{{\begin{tikzpicture}\ssnode{(0,0)};\ssnode{(.5,0)};\sstrans{(0,0)}{(.5,0)};\end{tikzpicture}}}
\newcommand{\anegde}{{\begin{tikzpicture}\ssnode{(0,0)};\ssnode{(.5,0)};\sstrans{(.5,0)}{(0,0)};\end{tikzpicture}}}
\begin{equation*}
  \vcenter{
    \xymatrix@R=2ex@C=2ex{
      &\anedge&&\anedge\\
      \anode\ar[drr]\ar[ur]&&\ar[ul]\anode\ar[ur]&&\ar[ul]\ar[dll]\anode\\
      &&\anedge\\
    }
  }
  \qquad\qquad\qquad\qquad
  \vcenter{
    \xymatrix@R=2ex@C=2ex{
      &2&&2\\
      1\ar[drr]\ar[ur]&&\ar[ul]1\ar[ur]&&\ar[ul]\ar[dll]1\\
      &&2\\
    }
  }
\end{equation*}
by drawing the graphs $N_I(0)$ and $N_I(1)$. Notice, that the object $3$ is the
colimit of the corresponding diagram in~$\linear$ (on the right), and this is
generally true for all objects of~$\linear$, moreover this diagram is described
by the functor $\El(N_I(3))\xrightarrow\pi\linear$. The notation $\El(P)$ refers
to the \emph{category of elements} of a presheaf $P\in\hat\C$, whose objects
are pairs $(A,p)$ with $A\in\C$ and $p\in P(A)$ and morphisms $f:(A,p)\to(B,q)$
are morphisms $f:A\to B$ in $\C$ such that $P(f)(q)=p$, and~$\pi$ is the first
projection functor. The functor $I:\graph\to\linear$ is thus a dense functor in
the sense of Definition~\ref{def:dense} below, see~\cite{mac1992sheaves} for
details.

\begin{proposition}
  \label{prop:realization}
  Given a functor $F:\C\to\D$, with $\D$ cocomplete, the associated
  nerve $N_F:\D\to\hat\C$ admits a left adjoint $R_F:\hat\C\to\D$
  called the \emph{realization} along~$F$. This functor is defined on
  objects~$P\in\hat\C$ by
  \[
  R_F(P)\qeq\colim(\El(P)\xrightarrow{\pi}\C\xrightarrow{F}\D)
  \]
\end{proposition}
\begin{proof}
  Given a presheaf~$P\in\hat\C$ and an object $D$, it can be checked directly
  that morphisms $P\to N_FD$ in~$\hat\C$ with cocones from
  $\El(P)\xrightarrow\D$ to $D$, which in turn are in bijection with morphisms
  $R_F(P)\to D$ in~$\D$, see~\cite{mac1992sheaves}.
\end{proof}

\begin{definition}
  \label{def:dense}
  A functor $F:\C\to\D$ is \emph{dense} if it satisfies one of the two
  equivalent conditions:
  \begin{enumerate}
  \item[(i)] the associated nerve functor $N_F:\D\to\hat\C$ is full
    and faithful,
  \item[(ii)] every object of $\D$ is canonically a colimit of objects in~$\C$:
    for every $D\in\D$,
    \begin{equation}
      \label{eq:dense-repr}
      D\qcong\colim(\El(N_FD)\xrightarrow{\pi}\C\xrightarrow{F}\D)
    \end{equation}
  \end{enumerate}
\end{definition}

\noindent




Since the functor $I$ is dense, every object of~$\linear$ is a finite colimit of
objects in $\graph$, and $\graph$ does not have any non-trivial colimit. One
could expect the free conservative finite cocompletion of $\linear$ to be the
free finite cocompletion $\concurrent$ of~$\graph$. We will see that this is not
the case because the image in $\linear$ of a non-trivial diagram in $\graph$
might still have a colimit. By Theorem~\ref{thm:cons-cocomp}, the
category~$\concurrent$ is the full subcategory of~$\hat\linear$ of presheaves
preserving limits, which we now describe explicitly. This category will turn out
to be equivalent to a full subcategory of~$\hat\G$ (Theorem~\ref{thm:fcfc}). We
should first remark that those presheaves satisfy the following properties:

\begin{proposition}
  \label{prop:cons-cocompl}
  Given a presheaf $P\in\hat\linear$ which is an object of $\concurrent$,
  \begin{enumerate}
  \item the underlying graph of $P$ is finite,
  \item for each non-empty path $x\twoheadrightarrow y$ there exists exactly one
    edge $x\to y$ (in particular there is at most one edge between two
    vertices),
  \item $P(n+1)$ is the set of paths of length~$n$ in the underlying graph
    of~$P$, and~$P(0)$ is reduced to one element.
  \end{enumerate}
\end{proposition}
\begin{proof}
  We suppose given a presheaf $P\in\concurrent$, it preserves limits by
  Theorem~\ref{thm:cons-cocomp}. The diagram on the left
  \[
  \vcenter{
    \xymatrix@R=1ex@C=2ex{
      &3&\\
      2\ar[ur]^{s^2_2}&&\ar[ul]_{s^2_0}2\\
      &\ar[ul]^{s^1_0}1\ar[ur]_{s^1_1}\\
    }
  }
  \qquad\qquad\qquad
  \vcenter{
    \xymatrix@R=1ex@C=2ex{
      &P(3)&\\
      P(2)\ar@{<-}[ur]^{P(s^2_2)}&&\ar@{<-}[ul]_{P(s^2_0)}P(2)\\
      &\ar@{<-}[ul]^{P(s^1_0)}P(1)\ar@{<-}[ur]_{P(s^1_1)}\\
    }
  }
  \]
  is a pushout in~$\linear$, or equivalently the dual diagram is a pullback in
  $\linear^\op$. Therefore, writing $D$ for the diagram
  $\vxym{2&\ar[l]_{s^1_0}1\ar[r]^{s^1_1}&2}$ in $\linear$, a presheaf
  $P\in\concurrent$ should satisfy $P((\colim D)^\op)\cong\lim P(D^\op)$, \ie
  the above pushout diagram in~$\linear$ should be transported by $P$ into the
  pullback diagram in~$\Set$ depicted on the right of the above figure. This
  condition can be summarized by saying that $P$ should satisfy the isomorphism
  $P(3)\cong P(2)\times_{P(1)}P(2)$ (and this isomorphism should respect obvious
  source and target maps given by the fact that the functor $P$ should send
   a limiting cone to a limiting cone). From this fact, one can
   deduce that the elements $\alpha$ of $P(3)$ are in bijection with the
  paths $x\to y\to z$ of length~$2$ in the underlying graph of $P$ going from
  $x=P(s^2_2s^1_1)(\alpha)$ to $z=P(s^2_0s^1_0)(\alpha)$.
  In particular, this implies that for any path $\alpha=x\to y\to z$ of length
  $2$ in the underlying graph of $P$, there exists an edge $x\to z$, which is
  $P(s^2_1)(\alpha)$. More generally, given any integer $n>1$, the object $n+1$
  is the colimit in $\linear$ of the diagram
  \begin{equation}
    \label{eq:pres-lin}
    \xymatrix@R=1ex@C=2.5ex{
      &2&&2&&&&2&&2\\
      1\ar[ur]^{s^1_1}&&\ar[ul]_{s^1_0}1\ar[ur]^{s^1_1}&&\ar[ul]_{s^1_0}&\ldots&\ar[ur]^{s^1_1}&&\ar[ul]_{s^1_0}1\ar[ur]^{s^1_1}&&\ar[ul]_{s^1_0}1\\
    }
  \end{equation}
  with $n+1$ occurrences of the object $1$, and $n$ occurrences of the object
  $2$. Therefore, for every $n\in\N$, $P(n+1)$ is isomorphic to the set of paths
  of length~$n$ in the underlying graph. Moreover, since the diagram
  \begin{equation}
    \label{eq:pres-trans}
    \xymatrix@R=1ex@C=2.3ex{
      &2&&2&&&&2&&2\\
      1\ar[drrrrr]_{s^1_1}\ar[ur]^{s^1_1}&&\ar[ul]_{s^1_0}1\ar[ur]^{s^1_1}&&\ar[ul]_{s^1_0}&\ldots&\ar[ur]^{s^1_1}&&\ar[ul]_{s^1_0}1\ar[ur]^{s^1_1}&&\ar[ul]_{s^1_0}\ar[dlllll]^{s^1_0}1\\
      &&&&&2\\
    }
  \end{equation}
  with $n+1$ occurrences of the object $1$ also admits the object $n+1$ as
  colimit, we should have
  $P(n+1)\cong P(n+1)\times P(2)$ between any two vertices~$x$ and~$y$, \ie for
  every non-empty path $x\twoheadrightarrow y$ there exists exactly one edge
  $x\to y$.  Also, since the object $0$ is initial in~$\linear$, it is the
  colimit of the empty diagram. The set~$P(0)$ should thus be the terminal set,
  \ie reduced to one element.
  %
  Finally, since~$I$ is dense, $P$ should be a finite colimit of the
  representables $N_I(1)$ and $N_I(2)$, the set $P(1)$ is necessarily finite, as
  well as the set $P(2)$ since there is at most one edge between two vertices.
\end{proof}

Conversely, we wish to show that the conditions mentioned in the above
proposition exactly characterize the presheaves in $\concurrent$ among those in
$\hat\linear$. In order to prove so, by Theorem~\ref{thm:cons-cocomp}, we
have to show that presheaves~$P$ satisfying these conditions preserve finite
limits in~$\linear$, \ie that for every finite diagram $D:\J\to\linear$
admitting a colimit we have $P(\colim D)\cong\lim(P\circ D^\op)$. It seems quite
difficult to characterize the diagrams admitting a colimit in~$\linear$, however
the following lemma shows that it is enough to check diagrams ``generated'' by a
graph which admits a colimit.

\begin{lemma}
  \label{lemma:pres-graph}
  A presheaf $P\in\hat\linear$ preserves finite limits if and only if it sends
  the colimits of diagrams of the form
  \begin{equation}
    \label{eq:diag-graph}
    \El(G)\xrightarrow{\pi_G}\graph\xrightarrow{I}\linear
  \end{equation}
  to limits in~$\Set$, where $G\in\hat\graph$ is a finite graph such that the
  above diagram admits a colimit. Such a diagram in~$\linear$ is said to be
  \emph{generated by the graph~$G$}.
\end{lemma}
\begin{proof}
  In order to check that a presheaf $P\in\hat\linear$ preserves finite limits,
  we have to check that it sends colimits of finite diagrams in~$\linear$
  \emph{which admit a colimit} to limits in~$\Set$, and therefore we have to
  characterize diagrams which admit colimits in~$\linear$. Suppose given a
  diagram $K:\J\to\linear$. Since~$I$ is dense, every object of linear is a
  colimit of a diagram involving only the objects~$1$ and~$2$ (see
  Definition~\ref{def:dense}). We can therefore suppose that this is the case in the
  diagram~$K$. Finally, it can be shown that diagram~$K$ admits the same
  colimits as a diagram containing only~$s_0^1$ and $s_1^1$ as arrows (these are
  the only non-trivial arrows in $\linear$ whose source and target are $1$ or
  $2$), in which every object~$2$ is the target of exactly one arrow $s_0^1$ and
  one arrow $s_1^1$. For instance, the diagram in~$\linear$ below on the left
  admits the same colimits as the diagram in the middle.
  \[
  \vxym{
    2&&3\\
    &\ar[ul]^{s_0^1}1\ar[dr]_{s_1^1}\ar[ur]^>>{s_2^2 s_1^1}&&\ar[dl]^{s_0^1}\ar[ul]_{s_0^2 s_0^1}1&\ar@/^5mm/[dll]^<<<<{s_0^1}1\\
    &&2
  }
  \qquad\qquad
  \vxym{
    &2&&2&&2&\\
    1\ar[ur]^>{s_1^1}&&1\ar[ul]_>{s_0^1}\ar[ur]^>{s_1^1}\ar[drr]_{s_1^1}&&1\ar[ul]_>{s_0^1}\ar[ur]^>{s_1^1}&&1\ar[dll]^{s_0^1}\ar[ul]_>{s_0^1}\\
    &&&&2&\\
  }
  \qquad\qquad
  \vxym{
    0\ar[r]&1\ar[r]\ar@/_/[rr]&2\ar[r]&3
  }
  \]
  Any such diagram~$K$ is obtained by gluing a finite number of diagrams of the
  form $\vxym{1\ar[r]^{s_1^1}&2&1\ar[l]_{s_0^1}}$ along objects~$1$, and is
  therefore of the form $\El(G)\xrightarrow\pi\graph\xrightarrow{I}\linear$ for
  some finite graph $G\in\hat\graph$: the objects of $G$ are the objects $1$ in
  $K$, the edges of $G$ are the objects $2$ in $K$ and the source and target of
  an edge~$2$ are respectively given by the sources of the corresponding arrows
  $s_1^1$ and $s_0^1$ admitting it as target. For instance, the diagram in the
  middle above is generated by the graph on the right.
  The fact that every diagram is generated by a presheaf (is a discrete
  fibration) also follows more abstractly and generally from the construction of
  the comprehensive factorization system
  on~$\Cat$~\cite{pare1973connected,street1973comprehensive}.
\end{proof}

Among diagrams generated by graphs, those admitting a colimit can be
characterized using the following proposition:

\begin{lemma}
  \label{lemma:diag-colim}
  Given a graph~$G\in\hat\graph$, the associated diagram~\eqref{eq:diag-graph}
  admits a colimit in~$\linear$ if and only if there exists $n\in\linear$ and a morphism
  $f:G\to N_In$ in $\hat\linear$ such that every morphism $g:G\to N_Im$ in
  $\hat\linear$, with $m\in\linear$, factorizes uniquely through~$N_In$:
  $
  \vxym{
    G\ar[r]^f\ar@/_/[rr]_g&N_In\ar@{.>}[r]&N_Im
  }
  $
\end{lemma}
\begin{proof}
  Follows from the existence of a partially defined left adjoint to $N_I$, in
  the sense of~\cite{pare1973connected}, given by the fact that $I$ is dense
  (see Definition~\ref{def:dense}).
\end{proof}

\noindent
We finally arrive at the following concrete characterization of diagrams admitting
colimits:

\begin{lemma}
  \label{lemma:graphs-with-colim}
  A finite graph~$G\in\hat\graph$ induces a diagram~\eqref{eq:diag-graph}
  in~$\linear$ which admits a colimit if and only if it is ``tree-shaped'', \ie
  it is
  \begin{enumerate}
  \item acyclic: for any vertex $x$, the only path $x\twoheadrightarrow x$ is
    the empty path,
  \item connected: for any pair of vertices $x$ and $y$ there exists a path
    $x\twoheadrightarrow y$ or a path $y\twoheadrightarrow x$.
  \end{enumerate}
\end{lemma}
\begin{proof}
  Given an object $n\in\linear$, recall that $N_In$ is the graph whose objects
  are elements of $[n]$ and there is an arrow $i\to j$ if and only if
  $i<j$. Given a finite graph $G$, morphisms $f:G\to N_In$ are therefore in
  bijection with functions $f:V_G\to[n]$, where $V_G$ denotes the set of
  vertices of~$G$, such that $f(x)<f(y)$ whenever there exists an edge $x\to y$
  (or equivalently,  there exists a non-empty path $x\twoheadrightarrow
  y$).

  Consider a finite graph~$G\in\hat\graph$, by Lemma~\ref{lemma:diag-colim}, it
  induces a diagram~\eqref{eq:diag-graph} admitting a colimit if there is a
  universal arrow $f:G\to N_In$ with $n\in\linear$. From this it follows that the
  graph is acyclic: otherwise, we would have a non-empty path
  $x\twoheadrightarrow x$ for some vertex~$x$, which would imply
  $f(x)<f(x)$. Similarly, suppose that $G$ is a graph with vertices~$x$ and~$y$
  such that there is no path $x\twoheadrightarrow y$ or $y\twoheadrightarrow x$,
  and there is an universal morphism $f:G\to N_In$ for some
  $n\in\linear$. Suppose that $f(x)\leq f(y)$ (the case where $f(y)\leq f(x)$ is
  similar). We can define a morphism $g:G\to N_I(n+1)$ by $g(z)=f(z)+1$ if there
  is a path $x\twoheadrightarrow z$, $g(y)=f(x)$ and $g(z)=f(z)$ otherwise. This
  morphism is easily checked to be well-defined.
  Since we always have $f(x)\leq f(y)$ and $g(x)>g(y)$, there is no morphism
  $h:N_In\to N_I(n+1)$ such that $h\circ f=g$.

  Conversely, given a finite acyclic connected graph~$G$, the relation $\leq$
  defined on morphisms by $x\leq y$ whenever there exists a path
  $x\twoheadrightarrow y$ is a total order. Writing $n$ for the number of
  vertices in~$G$, the function $f:G\to N_In$, which to a vertex associates the
  number of vertices strictly below it \wrt $\leq$, is universal in the sense of
  Lemma~\ref{lemma:diag-colim}.
\end{proof}


\begin{proposition}
  \label{prop:fcfc-linear}
  The free conservative finite cocompletion $\concurrent$ of $\linear$ is
  equivalent to the full subcategory of $\hat\linear$ whose objects are
  presheaves~$P$ satisfying the conditions of
  Proposition~\ref{prop:cons-cocompl}.
\end{proposition}
\begin{proof}
  By Lemma~\ref{lemma:pres-graph}, the category $\concurrent$ is equivalent to
  the full subcategory of~$\hat\linear$ whose objects are presheaves preserving
  limits of diagrams of the form~\eqref{eq:diag-graph} generated by some graph
  $G\in\hat\graph$ which admits a colimit, \ie by
  Lemma~\ref{lemma:graphs-with-colim} the finite graphs which are acyclic and
  connected. We write $G_n$ for the graph with~$[n]$ as vertices and edges
  $i\to(i+1)$ for $0\leq i<n-1$. It can be shown that any acyclic and connected
  finite graph can be obtained from the graph~$G_n$, for some $n\in\N$, by
  iteratively adding an edge $x\to y$ for some vertices~$x$ and~$y$ such that
  there exists a non-empty path~$x\twoheadrightarrow y$. Namely, suppose given
  an acyclic and connected finite graph~$G$. The relation $\leq$ on its
  vertices, defined by $x\leq y$ whenever there exists a path
  $x\twoheadrightarrow y$, is a total order, and therefore the graph~$G$
  contains $G_n$, where $n$ is the number of edges of~$G$. An edge in~$G$ which
  is not in~$G_n$ is necessarily of the form $x\to y$ with $x\leq y$, otherwise
  it would not be acyclic.
  Since by Proposition~\ref{prop:cons-cocompl}, see~\eqref{eq:pres-trans}, the
  diagram generated by a graph of the form
  
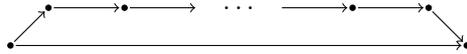
\[
  \begin{tikzpicture}[scale=1,baseline={(current bounding box.west)}]
    \ssnode{(0,0)};
    \sstrans{(0,0)}{(0.5,0.5)};
    \ssnode{(0.5,0.5)};
    \sstrans{(0.5,0.5)}{(1.5,0.5)};
    \ssnode{(1.5,0.5)};
    \sstrans{(1.5,0.5)}{(2.5,0.5)};
    \draw (3,0.5) node{$\ldots$};
    \sstrans{(3.5,0.5)}{(4.5,0.5)};
    \ssnode{(4.5,0.5)};
    \sstrans{(4.5,0.5)}{(5.5,0.5)};
    \ssnode{(5.5,0.5)};
    \sstrans{(5.5,0.5)}{(6,0)};
    \ssnode{(6,0)};
    \sstrans{(0,0)}{(6,0)};
  \end{tikzpicture}
  \]
  is preserved by presheaves in~$\concurrent$ (which corresponds to adding an
  edge between vertices at the source and target of a non-empty path), it is
  enough to show that presheaves in~$\concurrent$ preserve diagrams generated by
  graphs~$G_n$. This follows again by Proposition~\ref{prop:cons-cocompl},
  see~\eqref{eq:pres-lin}.
\end{proof}

One can notice that a presheaf $P\in\concurrent$ is characterized by its
underlying graph since $P(0)$ is reduced to one element and $P(n)$ with $n>2$ is
the set of paths of length~$n$ in this underlying graph: $P\cong I_*(I^*P)$. We
can therefore simplify the description of the cocompletion of~$\linear$ as
follows:

\begin{theorem}
  \label{thm:fcfc}
  The free conservative finite cocompletion $\concurrent$ of~$\linear$ is
  equivalent to the full subcategory of the category $\hat\graph$ of graphs,
  whose objects are finite graphs such that for every non-empty path
  $x\twoheadrightarrow y$ there exists exactly one edge $x\to y$. Equivalently,
  it can be described as the category whose objects are finite sets equipped
  with a transitive relation $<$, and functions respecting relations.
\end{theorem}

\noindent
In this category, pushouts can be explicitly described as follows:

\begin{proposition}
  \label{prop:pushouts}
  With the last above description, the pushout of a diagram in $\concurrent$
  $(B,<_B)\xleftarrow{f}(A,<_A)\xrightarrow{g}(C,<_C)$ is $B\uplus C/\sim$ with
  $B\owns b\sim c\in C$ whenever there exists $a\in A$ with $f(a)=b$ and
  $f(a)=c$, equipped with the transitive closure of the relation inherited by
  $<_B$ and $<_C$.
\end{proposition}

\noindent\emph{Lines with labels.}
The construction can be extended to the labeled case (\ie $L$ is not
necessarily a singleton). The forgetful functor $\hat\linear\to\Set$
sending a presheaf~$P$ to the set $P(1)$ admits a right adjoint
$\lss:\Set\to\hat\linear$. Given $n\in\N^*$ the elements of $\lss\lines(n)$ are
words~$u$ of length~$n$ over~$L$, with $!L(s_i^{n-1})(u)$ being the word
obtained from~$u$ by removing the $i$-th letter. The free conservative finite
cocompletion~$\concurrent$ of $\linear$ is the slice
category~$\linear/\lss\lines$, whose objects are pairs $(P,\ell)$ consisting of
a finite presheaf $P\in\hat\linear$ together with a \emph{labeling} morphism
$\ell:P\to\lss\lines$ of presheaves. Alternatively, the description of
Proposition~\ref{thm:fcfc} can be straightforwardly adapted by labeling the
elements of the objects by elements of~$\lines$ (labels should be preserved by
morphisms), thus justifying the use of labels for the vertices in following
examples.

\section{Examples}
\label{sec:examples}
In this section, we give some examples of merging (\ie pushout) of patches.

\begin{example}
  \label{ex:pushout}
  Suppose that starting from a file~$ab$, one user inserts a line~$a'$ at the
  beginning and~$c$ in the middle, while another one inserts a line~$d$ in the
  middle. After merging the two patches, the resulting file is the pushout of
  
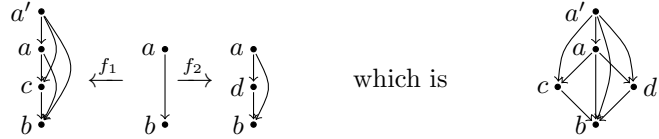
\[
  \begin{tikzpicture}[scale=1,baseline={(0,.5)}]
    \draw[->,shorten <=2pt,shorten >=2pt](0,1.5) .. controls (.5,.75) .. (0,0);
    \draw[->,shorten <=2pt,shorten >=2pt](0,1.5) .. controls (.25,1) .. (0,.5);
    \draw[->,shorten <=2pt,shorten >=2pt](0,1) .. controls (.25,.5) .. (0,0);
    \ssnode{(0,1.5)};
    \draw (0,1.5) node[left]{$a'$};
    \ssnode{(0,1)};
    \draw (0,1) node[left]{$a$};
    \ssnode{(0,.5)};
    \draw (0,.5) node[left]{$c$};
    \ssnode{(0,0)};
    \draw (0,0) node[left]{$b$};
    \sstrans{(0,1.5)}{(0,1)};
    \sstrans{(0,1)}{(0,.5)};
    \sstrans{(0,.5)}{(0,0)};
  \end{tikzpicture}
  \xleftarrow{f_1}
  \begin{tikzpicture}[scale=1,baseline={(0,.5)}]
    \ssnode{(0,1)};
    \draw (0,1) node[left]{$a$};
    \ssnode{(0,0)};
    \draw (0,0) node[left]{$b$};
    \sstrans{(0,1)}{(0,0)};
  \end{tikzpicture}
  \xrightarrow{f_2}
  \begin{tikzpicture}[scale=1,baseline={(0,.5)}]
    \draw[->,shorten <=2pt,shorten >=2pt](0,1) .. controls (.25,.5) .. (0,0);
    \ssnode{(0,1)};
    \draw (0,1) node[left]{$a$};
    \ssnode{(0,.5)};
    \draw (0,.5) node[left]{$d$};
    \ssnode{(0,0)};
    \draw (0,0) node[left]{$b$};
    \sstrans{(0,1)}{(0,.5)};
    \sstrans{(0,.5)}{(0,0)};
  \end{tikzpicture}
  \qquad\quad\text{which is}\qquad\quad
  \begin{tikzpicture}[scale=1,baseline={(0,.5)}]
    \ssnode{(.5,1.5)};
    \draw (.5,1.5) node[left]{$a'$};
    \ssnode{(.5,1)};
    \draw (.5,1) node[left]{$a$};
    \ssnode{(0,.5)};
    \draw (0,.5) node[left]{$c$};
    \ssnode{(1,.5)};
    \draw (1,.5) node[right]{$d$};
    \ssnode{(.5,0)};
    \draw (.5,0) node[left]{$b$};
    \sstrans{(.5,1.5)}{(.5,1)};
    \draw[->,shorten <=2pt,shorten >=2pt](.5,1.5) .. controls (.75,.75) .. (.5,0);
    \draw[->,shorten <=2pt,shorten >=2pt](.5,1.5) .. controls (.05,1) .. (0,.5);
    \draw[->,shorten <=2pt,shorten >=2pt](.5,1.5) .. controls (.95,1) .. (1,.5);
    \sstrans{(.5,1)}{(.5,0)};
    \sstrans{(.5,1)}{(0,.5)};
    \sstrans{(0,.5)}{(.5,0)};
    \sstrans{(.5,1)}{(1,.5)};
    \sstrans{(1,.5)}{(.5,0)};
  \end{tikzpicture}
  \]
\end{example}

\newcommand{\seq}{\textnormal{seq}}

\begin{example}
  Write $G_1$ for the graph with one vertex and no edges, and $G_2$ for the
  graph with two vertices and one edge between them. We write $s,t:G_1\to G_2$
  for the two morphisms in $\concurrent$. Since $\concurrent$ is finitely
  cocomplete, there is a coproduct $G_1+G_1$ which gives, by universal property,
  an arrow $\seq:G_1+G_1\to G_2$:
  \[
  \vxym{
    &G_2&\\
    G_1\ar[ur]^{s}\ar[r]&G_1+G_1\ar@{.>}[u]^<<{\seq}&\ar[l]\ar[ul]_tG_1\\
  }
  \qquad\text{or graphically}\qquad
  \vcenter
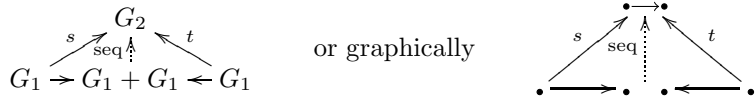
{
    \xymatrix{
      &\anedge&\\
      \anode\ar[ur]^-s\ar[r]&\twonodes\ar@{.>}[u]^<<<<<{\seq}&\ar[l]\ar[ul]_-t\anode
    }
  }
  \]
  that we call the \emph{sequentialization morphism}. This morphism corresponds
  to the following patch: given two possibilities for a line, a user can decide
  to turn them into two consecutive lines. We also write $\seq':G_1+G_1\to G_2$
  for the morphism obtained similarly by exchanging $s$ and $t$ in the above
  cocone. Now, the pushout of
  \[
  \anedge
  \quad\xleftarrow{\seq}\quad
  \twonodes
  \quad\xrightarrow{\seq'}\quad
  \anegde
  \qquad\text{is}\qquad
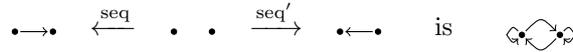
  \begin{tikzpicture}[baseline={(current bounding box.west)}]
    \ssnode{(0,0)};
    \ssnode{(.5,0)};
    \draw[->,shorten <=2pt,shorten >=2pt](0,0) .. controls (-.12,.12) .. (-.20,0) .. controls (-.12,-.12) .. (0,0);
    \draw[->,shorten <=2pt,shorten >=2pt](.5,0) .. controls (.62,.12) .. (.70,0) .. controls (.62,-.12) .. (.5,0);
    \draw[->,shorten <=2pt,shorten >=2pt](0,0) .. controls (.25,.25) .. (.5,0);
    \draw[->,shorten <=2pt,shorten >=2pt](.5,0) .. controls (.25,-.25) .. (0,0);
  \end{tikzpicture}
  \]
  which illustrates how cyclic graphs appear in~$\concurrent$ during the
  cocompletion of~$\linear$.
\end{example}

\begin{example}
  \label{ex:merging}
  With the notations of the previous example, by taking the coproduct of two
  copies of $\id_{G_1}:G_1\to G_1$, there is a universal morphism $G_1+G_1\to
  G_1$, which illustrates how two independent lines can be merged by a patch (in
  order to resolve conflicts).
  \[
  \vcenter
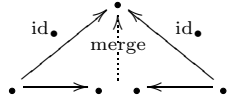
{
    \xymatrix{
      &\anode&\\
      \anode\ar[ur]^-{\id_\bullet}\ar[r]&\twonodes\ar@{.>}[u]|<<<<<{\mathrm{merge}}&\ar[l]\ar[ul]_-{\id_\bullet}\anode
    }
  }
  \]
\end{example}

\section{Handling deletions of lines}
\label{sec:extensions}
All the steps performed in previous sections in order to compute the free
conservative finite cocompletion of the category~$\tlinear$ can be adapted in
order to compute the cocompletion~$\concurrent$ of the category~$\linear$ as
introduced in Definition~\ref{def:linear}, thus adding support for deletion of
lines in patches.
In particular, the generalization of the description given by
Theorem~\ref{thm:fcfc} turns out to be as follows.

\begin{theorem}
  \label{theorem}
  The free conservative finite cocompletion~$\concurrent$ of the
  category~$\linear$ is the category whose objects are triples $(A,<,\ell)$
  where $A$ is a finite set of lines, $<$ is a transitive relation on $A$
  and~$\ell:A\to L$ associates a label to each line, and morphisms
  $f:(A,<_A,\ell_A)\to (B,<_B,\ell_B)$ are partial functions $f:A\to B$ such
  that for every $a,a'\in A$ both admitting an image under~$f$, we have
  $\ell_B(f(a))=\ell_A(a)$, and $a<_Aa'$ implies $f(a)<_Bf(a')$.
\end{theorem}

\noindent
Similarly, pushouts in this category can be computed as described in
Proposition~\ref{prop:pushouts}, generalized in the obvious way to partial
functions.

\begin{example}
  Suppose that starting from a file $abc$, one user inserts a line~$d$ after~$a$
  and the other one deletes the line~$b$. The merging of the two patches
  (in~$\concurrent'$) is the pushout of
  
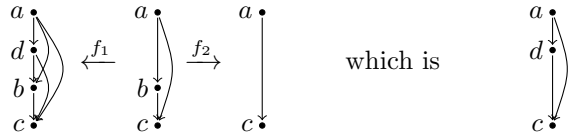
\[
  \begin{tikzpicture}[scale=1,baseline={(0,.75)}]
    \draw[->,shorten <=2pt,shorten >=2pt](0,1.5) .. controls (.5,.75) .. (0,0);
    \draw[->,shorten <=2pt,shorten >=2pt](0,1.5) .. controls (.25,1) .. (0,.5);
    \draw[->,shorten <=2pt,shorten >=2pt](0,1) .. controls (.25,.5) .. (0,0);
    \ssnode{(0,1.5)};
    \draw (0,1.5) node[left]{$a$};
    \ssnode{(0,1)};
    \draw (0,1) node[left]{$d$};
    \ssnode{(0,.5)};
    \draw (0,.5) node[left]{$b$};
    \ssnode{(0,0)};
    \draw (0,0) node[left]{$c$};
    \sstrans{(0,1.5)}{(0,1)};
    \sstrans{(0,1)}{(0,.5)};
    \sstrans{(0,.5)}{(0,0)};
  \end{tikzpicture}
  \xleftarrow{f_1}
  \begin{tikzpicture}[scale=1,baseline={(0,.75)}]
    \draw[->,shorten <=2pt,shorten >=2pt](0,1.5) .. controls (.25,.75) .. (0,0);
    \ssnode{(0,1.5)};
    \draw (0,1.5) node[left]{$a$};
    \ssnode{(0,.5)};
    \draw (0,.5) node[left]{$b$};
    \ssnode{(0,0)};
    \draw (0,0) node[left]{$c$};
    \sstrans{(0,1.5)}{(0,.5)};
    \sstrans{(0,.5)}{(0,0)};
  \end{tikzpicture}
  \xrightarrow{f_2}
  \begin{tikzpicture}[scale=1,baseline={(0,.75)}]
    \ssnode{(0,1.5)};
    \draw (0,1.5) node[left]{$a$};
    \ssnode{(0,0)};
    \draw (0,0) node[left]{$c$};
    \sstrans{(0,1.5)}{(0,0)};
  \end{tikzpicture}
  \qquad\quad\text{which is}\qquad\quad
  \begin{tikzpicture}[scale=1,baseline={(0,.75)}]
    \draw[->,shorten <=2pt,shorten >=2pt](0,1.5) .. controls (.25,.75) .. (0,0);
    \ssnode{(0,1.5)};
    \draw (0,1.5) node[left]{$a$};
    \ssnode{(0,1)};
    \draw (0,1) node[left]{$d$};
    \ssnode{(0,0)};
    \draw (0,0) node[left]{$c$};
    \sstrans{(0,1.5)}{(0,1)};
    \sstrans{(0,1)}{(0,0)};
  \end{tikzpicture}
  \]
  \ie the file $adc$. Notice that the morphism~$f_2$ is partial: $b$ has no
  image.
\end{example}

Interestingly, a category very similar to the one we have described in
Theorem~\ref{theorem} was independently proposed by Houston~\cite{bosker} based on
a construction performed in~\cite{cockett1995categories} for modeling
asynchronous processes. This category is not equivalent to ours because
morphisms are reversed partial functions: it is thus not the most general model
(in the sense of being the free finite cocompletion). As a simplified explanation for this,
consider the category $\FSet$ which is the finite cocompletion of~$\one$. This
category is finitely complete (in addition to cocomplete), thus $\FSet^\op$ is
finitely cocomplete and $\one$ embeds fully and faithfully in it. However,
$\FSet^\op$ is not the finite cocompletion of~$\one$. Another way to see this is
that this category does not contain the ``merging'' morphism of
Example~\ref{ex:merging}, but it contains a dual morphism ``duplicating'' lines.

\section{Concluding remarks and future works}
\label{sec:future-works}
In this paper, we have detailed how we could derive from universal constructions a
category which suitably models files resulting from conflicting modifications.
It is finitely cocomplete, thus the merging of any modifications of the file is
well-defined.

We believe that the interest of our 
 methodology lies in the fact that it
adapts easily to other more complicated base categories~$\linear$ than the two
investigated here: in future works, we should explain how to extend the
model in order to cope with multiple files (which can be moved, deleted, etc.),
different file types (containing text, or more structured data such as
\textsc{xml} trees). Also, the structure of repositories (partially ordered sets
of patches) is naturally modeled by event structures labeled by morphisms
in~$\concurrent$,
which will be detailed in future works, as well as
how to model usual operations on repositories: cherry-picking
(importing only one patch from another repository), using branches, removing a patch, etc. It
would also be interesting to explore axiomatically the addition of inverses for
patches, following other works hinted at in the introduction.

Once the theoretical setting is clearly established, we plan to investigate
algorithmic issues (in particular, how to efficiently represent and manipulate
the conflicting files, which are objects in~$\concurrent$). This should
eventually serve as a basis for the implementation of a theoretically sound and
complete distributed version control system (no unhandled corner-cases as in
most current implementations of \dvcs).

\medskip
\noindent{\emph{Acknowledgments.}}
The authors would like to thank P.-A.
Melliès, E. Haucourt, T. Heindel, T. Hirschowitz and the anonymous reviewers for their enlightening comments and
suggestions.

\bibliographystyle{abbrv}
\bibliography{patches}

\newpage
\appendix
\section{A geometric interpretation of presheaves on~$\tlinear$}

Since presheaf categories are sometimes a bit difficult to grasp, we recall here
the geometric interpretation that can be done for presheaves
in~$\hat\tlinear$. We forget about labels of lines and for simplicity suppose
that the empty file is not allowed (the objects are strictly positive
integers). In this section, we denote this category by~$\linear$. The same
reasoning can be performed on the usual category~$\tlinear$, and even $\linear$,
but the geometrical explanation is a bit more involved to describe.

In this case, the presheaves in~$\hat\linear$ can easily be described in
geometrical terms: the elements~$P$ of $\hat\linear$ are \emph{presimplicial
  sets}. Recall from Proposition~\ref{prop:presimp-free-cat} that the
category~$\linear$ is the free category whose objects are strictly positive
natural integers, containing for every integers~$n\in\N^*$ and $i\in[n+1]$
morphisms $s^n_i:n\to n+1$, subject to the relations
$s^{n+1}_is^n_j=s^{n+1}_{j+1}s^n_i$ whenever~$0\leq i\leq j<n$. Writing
$y:\linear\to\hat\linear$ for the Yoneda embedding, a representable presheaf
$y(n+1)\in\finhat\linear_+$ can be pictured geometrically as an $n$-simplex: a
$0$-simplex is a point, a $1$\nbd{}simplex is a segment, a $2$-simplex is a
(filled) triangle, a $3$-simplex is a (filled) tetrahedron, etc.:
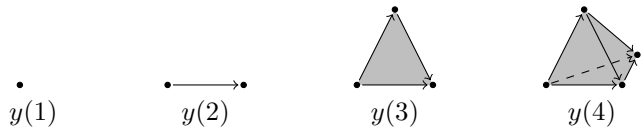
\[
\begin{array}{c@{\qquad\qquad}c@{\qquad\qquad}c@{\qquad\qquad}c@{\qquad\qquad}c}
  \begin{tikzpicture}
    \ssnode{(0,0)};
  \end{tikzpicture}
  \quad
  &
  \begin{tikzpicture}
    \ssnode{(0,0)};
    \ssnode{(1,0)};
    \sstrans{(0,0)}{(1,0)};
  \end{tikzpicture}
  &
  \begin{tikzpicture}
    \fill[fill=lightgray] (0,0) -- (1,0) -- (0.5,1) -- (0,0);
    \ssnode{(0,0)};
    \ssnode{(.5,1)};
    \ssnode{(1,0)};
    \sstrans{(0,0)}{(.5,1)};
    \sstrans{(.5,1)}{(1,0)};
    \sstrans{(0,0)}{(1,0)};
  \end{tikzpicture}
  &
  \begin{tikzpicture}
    \fill[fill=lightgray] (0,0) -- (1,0) -- (1.2,0.4) -- (0.5,1) -- (0,0);
    \sstransdashed{(0,0)}{(1.2,0.4)};
    \ssnode{(0,0)};
    \ssnode{(1,0)};
    \ssnode{(1.2,0.4)};
    \ssnode{(0.5,1)};
    \sstrans{(0,0)}{(1,0)};
    \sstrans{(0.5,1)}{(1,0)};
    \sstrans{(0,0)}{(0.5,1)};
    \sstrans{(1,0)}{(1.2,0.4)};
    \sstrans{(0.5,1)}{(1.2,0.4)};
  \end{tikzpicture}
  \\
  y(1)&y(2)&y(3)&y(4)
\end{array}
\]
Notice that the $n$-simplex has $n$ faces which are $(n-1)$\nbd{}dimensional
simplices, and these are given by the image under $y(s^n_i)$, with $i\in[n]$,
of the unique element of $y(n+1)(n+1)$: the $i$-th face of an $n$-simplex is the
$(n-1)$-simplex obtained by removing the $i$-th vertex from the simplex. More generally, a
\begin{wrapfigure}{r}{18mm}
  \vspace{-4ex}
  \[
  \begin{tikzpicture}
    \fill[fill=lightgray] (0,0) -- (1,0) -- (1,1) -- (0,0);
    \ssnode{(0,0)};
    \draw (0,0) node[below left]{$a$};
    \ssnode{(0,1)};
    \draw (0,1) node[above left]{$b$};
    \ssnode{(1,1)};
    \draw (1,1) node[above right]{$c$};
    \ssnode{(1,0)};
    \draw (1,0) node[below right]{$d$};
    \sstrans[{node[left]{$f$}}]{(0,0)}{(0,1)};
    \sstrans[{node[above]{$g$}}]{(0,1)}{(1,1)};
    \sstrans[{node[right]{$h$}}]{(1,0)}{(1,1)};
    \sstrans[{node[below]{$i$}}]{(0,0)}{(1,0)};
    \sstrans[{node{$j$}}]{(0,0)}{(1,1)};
    \draw (0.75,0.25) node{$\alpha$};
  \end{tikzpicture}
  \]
  \vspace{-7ex}
\end{wrapfigure}
presheaf $P\in\finhat\linear_+$ (a finite presimplicial set)
is a finite colimit of representables:
every such presheaf can be pictured as a gluing of simplices.
For instance, the half-filled square on the right corresponds to the
presimplicial set~$P$ with $P(1)=\set{a,b,c,d}$, $P(2)=\set{f,g,h,i,j}$,
$P(3)=\set{\alpha}$ with faces $P(s^1_1)(f)=a$, $P(s^1_0)(f)=b$, etc.

Similarly, in the labeled case, a labeled presheaf $(P,\ell)\in\linear/!$ can be
pictured as a presimplicial set whose vertices ($0$-simplices) are labeled by
elements
\begin{wrapfigure}{r}{15mm}
  \vspace{-4ex}
\[
\begin{tikzpicture}[scale=1.2]
  \fill[fill=lightgray] (0.5,0) -- (0,0.5) -- (0.5,1) -- cycle;
  \ssnode{(.5,1)};
  \draw (.5,1) node[right]{$a$};
  \ssnode{(0,.5)};
  \draw (0,.5) node[left]{$b$};
  \ssnode{(.5,0)};
  \draw (.5,0) node[right]{$c$};
  \sstrans[{node[above left=-0.1]{$ab$}}]{(.5,1)}{(0,.5)};
  \sstrans[{node[below left=-0.1]{$bc$}}]{(0,.5)}{(.5,0)};
  \sstrans[{node[right]{$ac$}}]{(.5,1)}{(.5,0)};
  \draw (0.25,0.5) node{\small $\ abc$};
\end{tikzpicture}
\]
\vspace{-7ex}
\end{wrapfigure}
of~$L$. The word labeling of higher-dimensional simplices can then be deduced
by concatenating the labels of the vertices it has as iterated faces. For
instance, an edge (a $1$-simplex) whose source is labeled by~$a$ and target is
labeled by~$b$ is necessarily labeled by the word $ab$, etc.

More generally, presheaves in $\tlinear$ can be pictured as \emph{augmented
  presimplicial sets} and presheaves in $\linear$ as \emph{augmented simplicial
  sets}, a description of those can for instance be found in Hatcher's book
\emph{Algebraic Topology}.

\section{Proofs of classical propositions}
\label{sec:classical-props}

In this section, we briefly recall proofs of well-known propositions as our
proofs rely on a fine understanding of those. We refer the reader
to~\cite{mac1992sheaves} for further details.

\begin{proposition*}{Proposition~\ref{prop:realization}}
  Given a functor $F:\C\to\D$, with $\D$ cocomplete, the associated
  nerve $N_F:\D\to\hat\C$ admits a left adjoint $R_F:\hat\C\to\D$
  called the \emph{realization} along~$F$. This functor is defined on
  objects~$P\in\hat\C$ by
  \[
  R_F(P)\qeq\colim(\El(P)\xrightarrow{\pi}\C\xrightarrow{F}\D)
  \]
\end{proposition*}
\begin{proof}
  In order to show the adjunction, we have to construct a natural family of
  isormorphisms $\D(R_F(P),D)\cong\hat\C(P,N_FD)$ indexed by a presheaf
  $P\in\hat\C$ and an object $D\in\D$. A natural transformation
  $\theta\in\hat\C(P,N_FD)$ is a family of functions $(\theta_C:PC\to
  N_FDC)_{C\in\C}$ such that for every morphism $f:C'\to C$ in~$\C$ the diagram
  \[
  \xymatrix@R=2ex@C=3ex{
    P(C)\ar[d]_{P(f)}\ar[r]^-{\theta_C}&\D(FC,D)\ar[d]^{\D(Ff,D)}\\
    P(C')\ar[r]_-{\theta_{C'}}&\D(FC',D)
  }
  \]
  commutes. It can also be seen as a family $(\theta_C(p):FC\to
  D)_{(C,p)\in\El(P)}$ of morphisms in $\D$ such that the diagram
  \[
  \vxym{
    FC\ar[dr]^{\theta_C(p)}\\
    &D\\
    \ar[uu]^{Ff}FC'\ar[ur]_{\theta_{C'}(P(f)(p))}\\
  }
  \quad\text{or equivalently}\quad
  \vxym{
    F\pi_P(C,p)\ar[dr]^>>>>{\theta_C(p)}\\
    &D\\
    F\pi_P(C',p')\ar[uu]^{F\pi_Pf}\ar[ur]_>>>>{\theta_{C'}(p')}
  }
  \]
  commutes for every morphism $f:C'\to C$ in~$\C$. This thus defines a cocone
  from $F\pi_P:\El(P)\to\D$ to $D$, and those cocones are in bijection with
  morphisms $R_F(P)\to D$ by definition of $R_F(P)$ as a colimit: we have shown
  $\D(R_F(P),D)\cong\hat\C(P,N_F(D))$, from which we conclude.
\end{proof}

The equivalence between the two conditions of Definition~\ref{def:dense} can be
shown as follows.

\begin{proposition}
  Given a functor $F:\C\to\D$, the two following conditions are equivalent:
  \begin{enumerate}
  \item[(i)] the associated nerve functor $N_F:\D\to\hat\C$ is full and
    faithful,
  \item[(ii)] every object of~$\D$ is canonically a colimit of objects in~$\C$:
    for every $D\in\D$,
    \begin{equation*}
      D\qcong\colim(\El(N_FD)\xrightarrow{\pi}\C\xrightarrow{F}\D)
    \end{equation*}
  \end{enumerate}
\end{proposition}
\begin{proof}
  In the case where~$\D$ is cocomplete, the nerve functor $N_F:\D\to\hat\C$ admits
  $R_F:\hat\C\to\D$ as right adjoint, and the equivalence amounts to showing
  that the right adjoint is full and faithful if and only if the counit is an
  isomorphism, which is a classical
  theorem~\cite[Theorem~IV.3.1]{maclane:cwm}. The construction can be adapted to
  the general case where~$\D$ is not necessarily cocomplete by considering
  $\colim(\El(-)\xrightarrow\pi\C\xrightarrow{F}\D):\hat\C\to\D$ as a partially
  defined left adjoint (see~\cite{pare1973connected}) and generalizing the
  theorem.
\end{proof}

\section{Proofs of the construction of the finite cocompletion}
\label{sec:cocomp-proofs}

\begin{proposition*}{Lemma~\ref{lemma:pres-graph}}
  A presheaf $P\in\hat\linear$ preserves finite limits, if and only if it sends
  the colimits of diagrams of the form
  \begin{equation*}
    \El(G)\xrightarrow{\pi_G}\graph\xrightarrow{I}\linear
  \end{equation*}
  to limits in~$\Set$, where $G\in\hat\graph$ is a finite graph such that the
  above diagram admits a colimit. Such a diagram in~$\linear$ is said to be
  \emph{generated by the graph~$G$}.
\end{proposition*}
\begin{proof}
  In order to check that a presheaf $P\in\hat\linear$ preserves finite limits,
  we have to check that it sends colimits of finite diagrams in~$\linear$
  \emph{which admit a colimit} to limits in~$\Set$, and therefore we have to
  characterize diagrams which admit colimits in~$\linear$. The number of
  diagrams to check can be reduced by using the facts that limits commute with
  limits~\cite{maclane:cwm}.
  For instance, the inclusion functor $I:\graph\to\linear$ is dense, which
  implies that every object $n\in\linear$ is canonically a colimit of the
  objects~$1$ and~$2$ by the formula
  $n\cong\colim(\El(N_In)\xrightarrow{\pi}\graph\xrightarrow{I}\linear)$, see
  Definition~\ref{def:dense}. Thus, given a finite diagram $K:\J\to\linear$, we
  can replace any object~$n$ different from~$1$ and~$2$ occurring in the diagram
  by the corresponding diagram
  $\El(N_In)\xrightarrow{\pi}\graph\xrightarrow{I}\linear$, thus obtaining a new
  diagram $K':\J\to\linear$ which admits the same colimit as~$K$.
  This shows that~$P$ will preserve finite limits if and only if it preserves
  limits of finite diagrams in~$\linear$ in which the only occurring objects
  are~$1$ and~$2$. Since the only non-trivial arrows in~$\linear$ between the
  objects $1$ and $2$ are $s_0^1,s_1^1:1\to 2$, and removing an identity arrow
  in a diagram does not change its colimit, the diagram $K$ can thus be
  assimilated to a bipartite graph with vertices labeled by $1$ or~$2$ and edges
  labeled by $s_0^1$ or~$s_1^1$, all edges going from vertices $1$ to vertices
  $2$.

  We can also reduce the number diagrams to check by remarking that some pairs
  of diagrams are ``equivalent'' in the sense that their image under~$P$ have
  the same limit, independently of~$P$. For instance, consider a diagram in
  which an object $2$ is the target of two arrows labeled by $s_0^1$ (on the
  left). The diagram obtained by identifying the two arrows along with the
  objects~$1$ in their source (on the right) can easily be checked to be
  equivalent by constructing a bijection between cocones of the first and
  cocones of the second.
  \[
  \vxym{
    \ar@{}[r]|{\ldots}2&2&&2&&\ar@{}[r]|{\ldots}2&2\\
    \ar@{.>}[urrr]\ar@{}[r]|{\ldots}1\ar@{.>}[u]\ar@{.>}[ur]&\ar@{.>}[u]\ar@{.>}[ul]1\ar@{.>}[urr]&\ar@{.>}[ull]\ar@{.>}[ul]1\ar[ur]_<<{s_0^1}&&1\ar[ul]^<<{s_0^1}\ar@{.>}[ur]\ar@{.>}[urr]&\ar@{.>}[ull]\ar@{}[r]|{\ldots}1\ar@{.>}[u]\ar@{.>}[ur]&\ar@{.>}[u]\ar@{.>}[ul]1\ar@{.>}[ulll]\\
  }
  \qquad\qquad\qquad
  \vxym{
    \ar@{}[r]|{\ldots}2&2&2&\ar@{}[r]|{\ldots}2&2&&\\
    \ar@{.>}[urr]\ar@{}[r]|{\ldots}1\ar@{.>}[u]\ar@{.>}[ur]&\ar@{.>}[u]\ar@{.>}[ul]1\ar@{.>}[ur]&\ar@{.>}[ull]\ar@{.>}[ul]1\ar[u]_{s_0^1}\ar@{.>}[ur]\ar@{.>}[urr]&\ar@{.>}[ul]\ar@{}[r]|{\ldots}1\ar@{.>}[u]\ar@{.>}[ur]&\ar@{.>}[u]\ar@{.>}[ul]1\ar@{.>}[ull]\\
  }
  \]
  More precisely, if we write $K:\J'\to\linear$ and $K:\J\to\linear$ for the two
  diagrams and $J:\J'\to\J$ for the obvious functor, the canonical arrow
  $\colim(K\circ J)\to\colim(K)$ is an isomorphism, \ie the functor $J$ is final.
  The same reasoning of course also holds with $s_1^1$ instead of $s_0^1$. We
  can therefore restrict ourselves to considering diagrams in which $2$ is the
  target of at most one arrow $s_0^1$, and of at most one arrow
  $s_1^1$. Conversely, if an object~$2$ is the target of no arrow $s_0^1$ (on
  the left), then we can add a new object $1$ and a new arrow from this object
  to the object $2$ (on the right) and obtain an equivalent diagram:
  \[
  \vxym{
    2&\ar@{}[r]|{\ldots}2&2&&\\
    &\ar@{.>}[ul]\ar@{}[r]|{\ldots}1\ar@{.>}[u]\ar@{.>}[ur]&\ar@{.>}[ull]1\ar@{.>}[u]\ar@{.>}[ul]\\
  }
  \qquad\quad
  \vxym{
    2&\ar@{}[r]|{\ldots}2&2&&\\
    1\ar[u]^{s_0^1}&\ar@{.>}[ul]\ar@{}[r]|{\ldots}1\ar@{.>}[u]\ar@{.>}[ur]&\ar@{.>}[ull]1\ar@{.>}[u]\ar@{.>}[ul]\\
  }
  \]
  The same reasoning holds with~$s_1^1$ instead of~$s_0^1$ and we can therefore
  restrict ourselves to diagrams in which every object~$2$ is the target of
  exactly one arrow~$s_0^1$ and one arrow~$s_1^1$.

  Any such diagram~$K$ is obtained by gluing a finite number of diagrams of the
  form
  \[
  \vxym{
    &2&\\
    1\ar[ur]^{s_1^1}&&1\ar[ul]_{s_0^1}\\
  }
  \]
  along objects~$1$, and is therefore of the form
  $\El(G)\xrightarrow\pi\graph\xrightarrow{I}\linear$ for some finite graph
  $G\in\hat\graph$: the objects of $G$ are the objects $1$ in $K$, the edges of
  $G$ are the objects $2$ in $K$ and the source and target of an edge~$2$ are
  respectively given by the sources of the corresponding arrows $s_1^1$ and $s_0^1$
  admitting it as target. For instance, the diagram on the left
  \[
  \vxym{
    &2&&2&&2&\\
    1\ar[ur]^>{s_1^1}&&1\ar[ul]_>{s_0^1}\ar[ur]^>{s_1^1}\ar[drr]_{s_1^1}&&1\ar[ul]_>{s_0^1}\ar[ur]^>{s_1^1}&&1\ar[dll]^{s_0^1}\ar[ul]_>{s_0^1}\\
    &&&&2&\\
  }
  \qquad\qquad\qquad
  \vxym{
    0\ar[r]&1\ar[r]\ar@/_/[rr]&2\ar[r]&3
  }
  \]
  is generated by the graph on the right.
\end{proof}

\begin{proposition*}{Lemma~\ref{lemma:diag-colim}}
  Given a graph~$G\in\hat\graph$, the associated diagram~\eqref{eq:diag-graph}
  admits a colimit in~$\linear$ if and only if there exists $n\in\linear$ and a morphism
  $f:G\to N_In$ in $\hat\linear$ such that every morphism $g:G\to N_Im$ in
  $\hat\linear$, with $m\in\linear$, factorizes uniquely through~$N_In$:
  \[
  \vxym{
    G\ar[r]^f\ar@/_/[rr]_g&N_In\ar@{.>}[r]&N_Im
  }
  \]
\end{proposition*}
\begin{proof}
  We have seen in proof of Proposition~\ref{prop:realization} that morphisms in
  $\hat\linear(G,N_In)$ are in bijection with cocones in~$\linear$ from
  $\El(G)\xrightarrow{\pi_G}\graph\xrightarrow{I}\linear$ to~$n$, and moreover
  given a morphism $h:n\to m$ in $\graph$ the morphism
  $\hat\linear(G,N_In)\to\hat\linear(G,N_Im)$ induced by post-composition with
  $N_Ih$ is easily checked to correspond to the usual notion of morphism between
  $n$\nbd{}cocones and $m$-cocones induced by $N_Ih$ (every morphism $N_In\to
  N_Im$ is of this form since $N_I$ is full and faithful). We can finally
  conclude using the universal property defining colimiting cocones.
\end{proof}

\section{Proofs for deletions of lines}
\label{sec:partial-proofs}

In this section, we detail proofs of properties mentioned in
Section~\ref{sec:extensions}.

\subsection{Sets and partial functions}
Before considering the conservative finite cocompletion of the
category~$\lines$, as introduced in Definition~\ref{def:linear}, it is
enlightening to study the category $\PSet$ of sets and partial functions. A
partial function $f:A\to B$ can always be seen
\begin{enumerate}
\item as a total function $f:A\to B\uplus\set{\bot_A}$ where $\bot_A$ is a fresh
  element \wrt $A$, where given $a\in A$, $f(a)=\bot_A$ means that the partial
  function is undefined on~$A$,
\item alternatively, as a total function $f:A\uplus\set{\bot_A}\to
  B\uplus\set{\bot_B}$ such that $f(\bot_A)=\bot_B$.
\end{enumerate}

\noindent
This thus suggests to consider the following category:

\begin{definition}
  The category~$\pSet$ of \emph{pointed sets} has pairs $(A,a)$ where $A$ is a
  set and $a\in A$ as objects, and morphisms $f:(A,a)\to(B,b)$ are (total)
  functions $f:A\to B$ such that $f(a)=b$.
\end{definition}

\noindent
Point (ii) of the preceding discussion can be summarized by saying that a
partial function can be seen as a pointed function and conversely:

\begin{proposition}
  The category~$\PSet$ of sets and partial functions is equivalent to the
  category~$\pSet$ of pointed sets.
\end{proposition}

It is easily shown that the forgetful functor $U:\pSet\to\Set$, sending a
pointed set $(A,a)$ to the underlying set~$A$, admits a left adjoint
$F:\Set\to\pSet$, defined on objects by $FA=(A\uplus\set{\bot_A},\bot_A)$. This
adjunction induces a monad $T=UF$ on~$\Set$, from which point (i) can be
formalized:

\begin{proposition}
  The category~$\PSet$ is equivalent to the Kleisli category~$\Set_T$ associated
  to the monad $T:\Set\to\Set$.
\end{proposition}


Finally, it turns out that the category~$\pSet$ of pointed sets might have been
discovered from~$\PSet$ using ``presheaf thinking'' as follows.
%
%
We write $\G$ for the full subcategory of~$\PSet$ containing two objects: the
empty set~$0=\emptyset$ and a set~$1=\set{\ast}$ with only one element, and two
non-trivial arrows $\star:0\to 1$ and $\bot:1\to 0$ (the undefined function)
such that $\bot\circ\star=\id_0$. We write $I:\G\to\PSet$ for the inclusion
functor. Consider the associated nerve functor $N_I:\PSet\to\hat\G$. Given a set
$A$ the presheaf $N_IA\in\hat\G$ is such that:
\begin{itemize}
\item $N_IA0=\PSet(I0,A)\cong\set{\star}$: the only morphism $0\to A$ in~$\PSet$
  is noted $\star$,
\item $N_IA1=\PSet(I1,A)\cong A\uplus\set{\bot_A}$: a morphism $1\to A$ is
  characterized by the image of $*\in A$ which is either an element of $A$ or
  undefined,
\item $N_IA\star:N_IA1\to N_IA0$ is the constant function whose image is
  $\star$,
\item $N_IA\bot:N_IA0\to N_IA1$ is the function such that the image of $\star$
  is $\bot_A$.
\end{itemize}
Moreover, given $A,B\in\PSet$ a natural transformation from $N_IA$ to $N_IB$ is
a pair of functions $f:A\uplus\set{\bot_A}\to B\uplus\set{\bot_B}$ and
$g:\set{\star}\to\set{\star}$ such that the diagrams
\[
\vxym{
  A\uplus\set{\bot_A}\ar[r]^f\ar[d]_{N_IA\star}&\ar[d]^{N_IB\star}B\uplus\set{\bot_B}\\
  \set{\star}\ar[r]_g&\set{\star}
}
\qquad\text{and}\qquad
\vxym{
  A\uplus\set{\bot_A}\ar[r]^f&\uplus\set{\bot_B}\\
  \set{\star}\ar[r]_g\ar[u]^{N_IA\bot}&\ar[u]_{N_IB\bot}\set{\star}
}
\]
commutes. Since $\set\star$ is the terminal set, such a natural transformation
is characterized by a function $f:A\uplus\set{\bot_A}\to B\uplus\set{\bot_B}$
such that $f(\bot_A)=\bot_B$. The functor $N_I:\PSet\to\hat\G$ is thus dense and
its image is equivalent to $\pSet$.

\subsection{A cocompletion of~$\linear$}
The situation with regards to the category~$\linear$ is very similar. We follow
the plan of Section~\ref{sec:cocompl} and first investigate the unlabeled case:
$\linear$ is the category with integers as objects and partial injective
increasing functions $f:[m]\to[n]$ as morphisms $f:m\to n$.

We write~$\G$ for the full subcategory of $\linear$ whose objects are $0$, $1$
and $2$. This is the free category on the graph
\[
\vxym{
  0\ar@<+.5ex>[rr]^{s_0^0}&&\ar@<+.5ex>[ll]^{d_0^0}1\ar@<+2.5ex>[rr]^{s^1_0}\ar@<+1ex>[rr]|{s^1_1}&&\ar@<+1ex>[ll]|{d^1_0}\ar@<+2.5ex>[ll]^{d^1_1}2
}
\]
subject to the relations
\begin{equation}
  \label{eq:g-ax}
  s^1_0s^0_0=s^1_1s^0_0
  \quad
  d^0_0s^0_0=\id_1
  \quad
  d^1_0s^1_0=\id_2
  \quad
  d^1_1s^1_1=\id_2
  \quad
  d^0_0d^1_0=d^0_0d^1_1
\end{equation}
(see Proposition~\ref{prop:presimp-free-cat}). We write $I:\G\to\linear$ for the
embedding and consider the associated nerve
functor~$N_I:\linear\to\hat\G$. Suppose given an object $n\in\linear$, the
associated presheaf $N_In$ can be described as follows. Its sets are
\begin{itemize}
\item $N_In0=\linear(I0,n)\cong\set{\star}$,
\item $N_In1=\linear(I1,n)\cong [n]\uplus\set{\bot}$,
\item $N_In2=\linear(I2,n)\cong\\
  \setof{(i,j)\in[n]\times [n]}{i<j}\uplus\setof{(\bot,i)}{i\in[n]}\uplus\setof{(i,\bot)}{i\in[n]}\uplus\set{(\bot,\bot)}$:
  a partial function $f:2\to n$ is characterized by the pair of images
  $(f(0),f(1))$ of $0,1\in[n]$, where $\bot$ means undefined.
\end{itemize}
and morphisms are
\begin{itemize}
\item $N_Ins_0^0:N_In1\to N_In0$ is the constant function whose image is $\star$,
\item $N_Ind_0^0:N_In0\to N_In1$ is the function whose image is $\bot$,
\item $N_Ins_0^1:N_In2\to N_In1$ is the second projection,
\item $N_Ins_1^1:N_In2\to N_In1$ is the first projection,
\item $N_Ind^1_0:N_In1\to N_In 2$ sends $i\in[n]\uplus\set{\bot}$ to $(\bot,i)$
\item $N_Ind^1_1:N_In1\to N_In 2$ sends $i\in[n]\uplus\set{\bot}$ to $(i,\bot)$
\end{itemize}
Such a presheaf can be pictured as a graph with $N_In1$ as set of vertices,
$N_In2$ as set of edges, source and target being respectively given by the
functions $N_Ins_1^1$ and $N_Ins_0^1$:
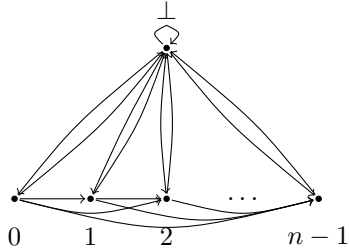
\[
\begin{tikzpicture}[scale=1,baseline={(0,.75)}]
  \ssnode{(2,2)};
  \draw (2,2.5) node {$\bot$};
  \draw[->,shorten <=2pt,shorten >=2pt](2,2) .. controls (1.8,2.2) .. (2,2.3) .. controls (2.2,2.2) .. (2,2);
  \draw[->,shorten <=2pt,shorten >=2pt](2,2) .. controls (.8,1) .. (0,0);
  \draw[->,shorten <=2pt,shorten >=2pt](0,0) .. controls (1.2,1) .. (2,2);
  \draw[->,shorten <=2pt,shorten >=2pt](2,2) .. controls (1.45,1) .. (1,0);
  \draw[->,shorten <=2pt,shorten >=2pt](1,0) .. controls (1.65,1) .. (2,2);
  \draw[->,shorten <=2pt,shorten >=2pt](2,2) .. controls (1.9,1) .. (2,0);
  \draw[->,shorten <=2pt,shorten >=2pt](2,0) .. controls (2.1,1) .. (2,2);
  \draw[->,shorten <=2pt,shorten >=2pt](2,2) .. controls (2.8,1) .. (4,0);
  \draw[->,shorten <=2pt,shorten >=2pt](4,0) .. controls (3.2,1) .. (2,2);
  \ssnode{(0,0)};
  \draw (0,-.5) node{$0$};
  \ssnode{(1,0)};
  \draw (1,-.5) node{$1$};
  \ssnode{(2,0)};
  \draw (2,-.5) node{$2$};
  \draw (3,0) node {$\ldots$};
  \ssnode{(4,0)};
  \draw (4,-.5) node{$n-1$};
  \sstrans{(0,0)}{(1,0)};
  \sstrans{(1,0)}{(2,0)};
  \draw[->,shorten <=2pt,shorten >=2pt](0,0) .. controls (1,-.25) .. (2,0);
  \draw[->,shorten <=2pt,shorten >=2pt](0,0) .. controls (2,-.5) .. (4,0);
  \draw[->,shorten <=2pt,shorten >=2pt](1,0) .. controls (2.5,-.37) .. (4,0);
  \draw[->,shorten <=2pt,shorten >=2pt](2,0) .. controls (3,-.25) .. (4,0);
\end{tikzpicture}
\]
Its vertices are elements of $[n]\uplus\set{\bot}$ and edges are of the form
\begin{itemize}
\item $i\to j$ with $i,j\in[n]$ such that $i<j$,
\item $i\to\bot$ for $i\in[n]$
\item $\bot\to i$ for $i\in[n]$
\item $\bot\to\bot$
\end{itemize}
Morphisms are usual graphs morphisms which preserve the vertex~$\bot$. We are
thus naturally lead to define the following categories of pointed graphs and
graphs with partial functions. We recall that a graph $G=(V,s,t,E)$ consists of
a set $V$ of vertices, a set $E$ of edges and two functions $s,t:E\to V$
associating to each edge its source and its target respectively.

\begin{definition}
  We define the category $\pGraph$ of \emph{pointed graphs} as the category
  whose objects are pairs $(G,x)$ with $G=(V,E)$ and $x\in V$ such that for
  every vertex there is exactly one edge from and to the distinguished
  vertex~$x$, and morphisms~$f:G\to G'$ are usual graph morphisms consisting of
  a pair $(f_V,f_E)$ of functions $f_V:V_G\to V_{G'}$ and $f_E:E_G\to E_{G'}$
  such that for every edge $e\in E_G$, $f_V(s(e))=s(f_E(e))$ and
  $f_V(t(e))=t(f_E(e))$, which are such that the distinguished vertex is
  preserved by $f_V$.
\end{definition}

\begin{definition}
  We define the category $\PGraph$ of \emph{graphs and partial morphisms} as the
  category whose objects are graphs and morphisms $f:G\to G'$ are pairs
  $(f_V,f_E)$ of partial functions $f_V:V_G\to V_{G'}$ and $f_E:E_G\to E_{G'}$
  such that
  \begin{itemize}
  \item for every edge $e\in E_G$ such that $f_E(e)$ is defined, $f_V(s(e))$ and
    $f_V(t(e))$ are both defined and satisfy $f_V(s(e))=s(f_E(e))$ and
    $f_V(t(e))=t(f_E(e))$,
  \item for every edge $e\in E_G$ such that $f_V(s(e))$ and $f_V(t(e))$ are both
    defined, $f_E(e)$ is also defined.
  \end{itemize}
  More briefly: a morphism is defined on an edge if and only it is defined on
  its source and on its target.
\end{definition}

\noindent
Similarly to previous section, a partial morphism of graph can be seen as a
pointed morphism of graph and conversely:

\begin{proposition}
  \label{prop:pgraph-eq}
  The categories $\pGraph$ and $\PGraph$ are equivalent.
\end{proposition}

\noindent
Now, notice that the category~$\linear$ is isomorphic to the full subcategory of
$\PGraph$ whose objects are the graphs whose set of objects is $[n]$ for some
$n\in\N$, and such that there is an edge $i\to j$ precisely when $i<j$. Also
notice that the full subcategory of $\pGraph$ whose objects are the
graphs~$N_In$ (with~$\bot$ as distinguished vertex) with~$n\in\N$ is isomorphic
to the full subcategory of $\hat\G$ whose objects are the $N_In$ with
$n\in\N$. And finally, the two categories are equivalent via the isomorphism of
Proposition~\ref{prop:pgraph-eq}. From this, we immediately deduce that the
functor $N_I:\linear\to\hat\G$ is full and faithful, \ie

\begin{proposition}
  The functor $I:\graph\to\linear$ is dense.
\end{proposition}

We can now follow Section~\ref{sec:cocompl} step by step, adapting each
proposition as necessary. The conditions satisfied by presheaves
in~$\concurrent$ introduced in Proposition~\ref{prop:cons-cocompl} are still
valid in our new case:

\begin{proposition}
  \label{prop:plinear-cond}
  Given a presheaf $P\in\hat\linear$ which is an object of $\concurrent$,
  \begin{enumerate}
  \item the underlying graph of $P$ is finite,
  \item for each non-empty path $x\twoheadrightarrow y$ there exists exactly one
    edge $x\to y$,
  \item $P(n+1)$ is the set of paths of length~$n$ in the underlying graph
    of~$P$, and~$P(0)$ is reduced to one element.
  \end{enumerate}
\end{proposition}
\begin{proof}
  The diagrams of the form~\eqref{eq:pres-lin} and~\eqref{eq:pres-trans} used in
  proof of Proposition~\ref{prop:cons-cocompl} still admit the same colimit
  $n+1$ with the new definition of~$\linear$ and $0$ is still initial. It can be
  checked that the limit of the image under a presheaf $P\in\hat\linear$ of a
  diagram~\eqref{eq:pres-lin} is still the set of paths of length~$n$ in the
  underlying graph of~$P$.
\end{proof}

\noindent
Lemma~\ref{lemma:pres-graph} is also still valid:

\begin{lemma}
  A presheaf $P\in\hat\linear$ preserves finite limits, if and only if it sends
  the colimits of diagrams of the form
  \begin{equation*}
    \El(G)\xrightarrow{\pi_G}\graph\xrightarrow{I}\linear
  \end{equation*}
  to limits in~$\Set$, where $G\in\hat\graph$ is a finite pointed graph such
  that the above diagram admits a colimit. Such a diagram in~$\linear$ is said
  to be \emph{generated by the pointed graph~$G$}.
\end{lemma}
\begin{proof}
  The proof of Lemma~\ref{lemma:pres-graph} was done ``by hand'', but we
  mentioned a more abstract alternative proof. In the present case, a similar
  proof can be done but would be really tedious, so we provide the abstract
  one. In order to illustrate why we have to do so, one can consider the
  category of elements associated to the presheaves representable by $0$ and
  $1$, which are clearly much bigger than in the case of
  Section~\ref{sec:cocompl}:
  \[
  \El(N_I0)
  \qcong
  \vcenter{
    \xymatrix@C=10ex{
      \star\ar@<+.5ex>[r]^{s_0^0}&\ar@<+.5ex>[l]^{d_0^0}\bot\ar@<+2.5ex>[r]^{s^1_0}\ar@<+1ex>[r]|{s^1_1}&\ar@<+1ex>[l]|{d^1_0}\ar@<+2.5ex>[l]^{d^1_1}(\bot,\bot)
    }
  }
  \]
  and
  \[
  \El(N_I1)
  \qcong
  \vcenter{
    \xymatrix@C=10ex{
      &&\ar@<-.5ex>[dl]|{d^1_0}\ar@<+1ex>[dl]|{d^1_1}(\bot,\bot)\\
      &\ar@<+.5ex>[dl]|{d_0^0}\bot\ar@<+4ex>[ur]^{s^1_0}\ar@<+2.5ex>[ur]|{s^1_1}\ar[r]|{s^1_1}\ar[dr]|>>>>>>>>{s^1_0}&\ar@<+.5ex>[dl]|>>>>>>>>>>{d^1_0}(\bot,1)\\
      \star\ar@<+.5ex>[ur]^{s_0^0}\ar[r]_{s^0_0}&1\ar@<+.5ex>[ur]|>>>>>>>>{s^1_0}\ar@<+.5ex>[r]|{s^1_1}&\ar@<+.5ex>[l]^{d^1_1}(1,\bot)
    }
  }
  \]
  subject to relations which follow directly from~\eqref{eq:g-ax}.

  Before going on with the proof, we need to introduce a few notions. A functor
  $F:\C\to\D$ is called \emph{final} if for every category~$\E$ and diagram
  $G:\D\to\E$ the canonical morphism $\colim(G\circ F)\to\colim(G)$ is an
  isomorphism~\cite{maclane:cwm}: restricting a diagram along $F$ does not
  change its colimit. Alternatively, these functors can be characterized as
  functors such that for every object $D\in\D$ the category is non-empty and
  connected (there is a zig-zag of morphisms between any two objects). A functor
  $F:\C\to\D$ is called a \emph{discrete fibration} if for any object $C\in\C$
  and morphism $g:D\to FC$ in~$\D$ there exists a unique morphism $f:C'\to C$
  in~$\C$ such that $Ff=g$ called the lifting of~$g$. To any such discrete
  fibration one can associate a presheaf~$P\in\hat\D$ defined on any $D\in\D$ by
  $PD=F^{-1}(D)=\setof{C\in\C}{FC=D}$ and on morphisms~$g:D'\to D$ as the
  function $Pg$ which to $C\in PD$ associates the source of the lifting of $g$
  with codomain~$C$. Conversely, any presheaf $P\in\hat\D$ induces a discrete
  fibration $\El(P)\xrightarrow\pi\D$, and these two operations induce an
  equivalence of categories between the category~$\hat\D$ and the category of
  discrete fibrations over~$\D$. It was shown by Paré, Street and
  Walters~\cite{pare1973connected,street1973comprehensive} that any functor
  $F:\C\to\D$ factorizes as final functor $J:\C\to\E$ followed by a discrete
  fibration $K:\E\to\D$, and this factorization is essentially unique: this is
  called the \emph{comprehensive factorization} of a functor. More explicitly,
  the functor~$K$ can be defined as follows. The inclusion functor $\Set\to\Cat$
  which send a set to the corresponding discrete category admits a left adjoint
  $\Pi_0:\Cat\to\Set$, sending a category to its connected components (its set
  of objects quotiented by the relation identifying two objects linked by a
  zig-zag of morphisms). The discrete fibration part~$K$ above can be defined as
  $\El(P)\xrightarrow\pi\D$ where $P\in\hat\D$ is the presheaf defined by
  $P=\Pi_0(-/F)$. In this precise sense, every diagram~$F$ in $\D$ is
  ``equivalent'' to one which is ``generated'' by a presheaf~$P$ on~$\D$ (we
  adopted this informal terminology in the article in order to avoid having to
  introduce too many categorical notions).

  In our case, we can thus restrict to diagrams in~$\linear$ generated by
  presheaves on $\linear$. Finally, since $I:\graph\to\linear$ is dense, we can
  further restrict to diagrams generated by presheaves on~$\graph$ by
  interchange of colimits.
\end{proof}

\noindent
Lemma~\ref{lemma:graphs-with-colim} applies almost as in
Section~\ref{sec:cocompl}: since the morphism $f:G\to N_In$ (seen as a partial
functions between graphs) has to satisfy the universal property of
Lemma~\ref{lemma:diag-colim}, by choosing for every vertex~$x$ of~$G$ a partial
function $g_x:G\to N_Im$ which is defined on~$x$ (such a partial function always
exists), it can be shown that the function $f$ has to be total. The rest of the
proof can be kept unchanged. Similarly, Proposition~\ref{prop:fcfc-linear}
applies with proof unchanged.

Finally, we have that
\begin{theorem}
  The free conservative finite cocompletion $\concurrent$ of $\linear$ is
  equivalent to the full subcategory of $\hat\linear$ whose objects are
  presheaves~$P$ satisfying the conditions of
  Proposition~\ref{prop:plinear-cond}. Since its objects~$P$ satisfy
  $I_*I^*(P)\cong P$, it can equivalently be characterized as the full
  subcategory of $\hat\graph$ whose objects~$P$ are
  \begin{enumerate}
  \item \emph{finite},
  \item \emph{transitive}: for each non-empty path $x\twoheadrightarrow y$ there
    exists exactly one edge $x\to y$,
  \item \emph{pointed}: $P(0)$ is reduced to one element.
  \end{enumerate}
\end{theorem}

\noindent
From this characterization (which can easily be extended to the labeled case),
along with the correspondence between pointed graphs and graphs with partial
functions (Proposition~\ref{prop:pgraph-eq}), the category is shown to be
equivalent to the category described in Theorem~\ref{theorem}: the relation is
defined on vertices $x,y$ of a graph~$G$ by $x<y$ whenever there exists a path
$x\twoheadrightarrow y$.

As in case of previous section, the forgetful functor $\pGraph\to\Graph$ admits
a left adjoint, thus inducing a monad on $\Graph$. The category $\pGraph$ is
equivalent to the Kleisli category associated to this monad, which is closely
related to the \emph{exception monad} as discussed in~\cite{bosker}.

\section{Modeling repositories}
\label{sec:repos}
We briefly detail here the modeling of repositories evoked in
Section~\ref{sec:future-works}. As explained in the introduction, repositories
can be modeled as partially ordered sets of patches, \ie morphisms
in~$\linear$. Since some of them can be incompatible, it is natural to model
them as particular labeled event structures.

\begin{definition}
  An \emph{event structure} $(E,\leq,\#)$ consists of a set $E$ of
  \emph{events}, a partial order relation~$\leq$ on $E$ and
  \emph{incompatibility} relation on events.
  We require that
  \begin{enumerate}
  \item for any event $e$, the downward closure of $\set{e}$ is finite and
  \item given $e_1$, $e_1'$ and $e_2$ such that $e_1\leq e_1'$ and $e_1\# e_2$,
    we have $e_1'\#e_2$.
  \end{enumerate}
\end{definition}
Two events $e_1$ and $e_2$ are \emph{compatible} when they are not incompatible,
and \emph{independent} when they are compatible and neither $e_1\leq e_2$ nor
$e_2\leq e_1$. A \emph{configuration} $x$ is a finite downward-closed set of
compatible events. An event $e_2$ is a \emph{successor} of an event $e_1$ when
$e_1\leq e_2$ and there is no event in between. Given an event $e$ we write $\da
e$ for the configuration, called the \emph{cause} of $e$, obtained as the
downward closure of $\set{e}$ from which~$e$ was removed. A \emph{morphism} of
event structures $f:(E,\leq,\#)\to(E',\leq',\#')$ is an injective function
$f:E\to E'$ such that the image of a configuration is a configuration. We
write~$\ES$ for the category of event structures.

To every event structure $E$, we can associate a \emph{trace graph}~$T(E)$ whose
vertices are configurations and edges are of the form
$x\xrightarrow{e}x\uplus\set{e}$ where $x$ is a configuration such that
$e\not\in x$ and $x\uplus\set{e}$ is a configuration. A \emph{trace} is a path
$x\twoheadrightarrow y$ in this graph. Notice that two paths
$x\twoheadrightarrow y$ are of the same length. Moreover, given two
configurations $x$ and $y$ such that $x\subseteq y$, there exists necessarily a
path~$x\twoheadrightarrow y$. It can be shown that this operation provides a
faithful embedding $T:\ES\to\Graph$ from the category of event structures to the
category of graphs, which admits a right
adjoint.

\begin{example}
  An event structure with five events is pictured on the left (arrows represent
  causal dependencies and $\sim$ incompatibilities). The associated trace graph
  is pictured on the right.
  \[
  \vxym{
    &d\ar@{~}[drr]&\\
    b\ar[ur]&&\ar[ul]c\ar@{~}[r]&c'\\
    &\ar[ul]a\ar[ur]\ar[urr]&\\
  }
  \hspace{12ex}
  {
    \vcenter
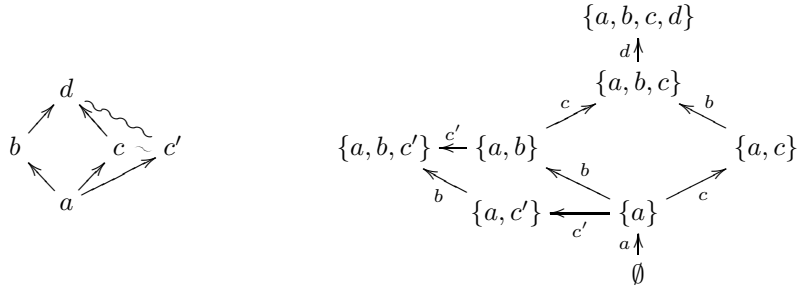
{
      \xymatrix@C=2ex@R=2ex{
        &&\set{a,b,c,d}\\
        &&\ar[u]^-d\set{a,b,c}\\
        \set{a,b,c'}&\ar[l]_-{c'}\ar[ur]^-c\set{a,b}&&\ar[ul]_-b\set{a,c}\\
        &\ar[ul]^-b\set{a,c'}&\ar[l]^-{c'}\ar[ul]_-b\ar[ur]_-c\set{a}&\\
        &&\emptyset\ar[u]^-a&\\
      }
    }
  }
  \]
\end{example}


\begin{definition}
  A \emph{categorical event structure} $(E,\lambda)$ in a category $\C$ with an
  initial object consists of an event structure equipped with a \emph{labeling}
  functor $\lambda:(TE)^*\to\C$, where $(TE)^*$ is the free category generated
  by the graph $TE$, such that $\lambda\emptyset$ is the initial object of~$\C$
  and the image under~$\lambda$ of every square
  \[
  \vxym{
    &z&\\
    y_1\ar[ur]^{e_2}&&\ar[ul]_{e_1}y_2\\
    &\ar[ul]^{e_1}x\ar[ur]_{e_2}&\\
  }
  \]
  in $TE$ is a pushout in~$\C$.
\end{definition}

The following proposition shows that a categorical event structure is
characterized by a suitable labeling of events of~$E$ by morphisms of~$\C$.

\begin{proposition}
  The functor $\lambda$ is characterized, up to isomorphism, by the image of
  the transitions $\da e\xrightarrow{e}\da e\uplus\set{e}$.
\end{proposition}

We can now define a \emph{repository} to be simply a finite categorical event
structure $(E,\leq,\#,\lambda:T(E)\to\linear)$. Such a repository extends to a
categorical event structure $(E,\leq,\#_0,I\circ\lambda:T(E)\to\concurrent)$,
where $\#_0$ is the empty conflict relation. The \emph{state}~$S$ of such an
event structure is the file obtained as the image $S=I\circ\lambda(E)$ of the
maximal configuration: this is the file that the users is currently editing
given his repository. Usual operations on repositories can be modeled in this
context, for instance importing the patches of another repository is obtained by
a pushout construction (the category of repositories is finitely cocomplete).

\end{document}